\documentclass[11pt]{article}

\usepackage[utf8]{inputenc}
 \usepackage{amsthm}
 \usepackage{mathtools}
\usepackage{amsfonts} 
\usepackage{amssymb} 
\usepackage[top=3cm,inner=3.5cm,outer=3.5cm,bottom=4cm]{geometry}
 \usepackage{setspace}
  \usepackage{mathrsfs}
\usepackage[usenames,dvipsnames]{xcolor}
\usepackage[pageanchor=false,colorlinks=true,linkcolor=blue,citecolor=blue,urlcolor=blue]{hyperref}

\newcommand\bR{\mathbb{R}}

\newcommand\bC{\mathbb{C}}

\newcommand\bT{\mathbb{T}}

\newcommand\sM{\mathcal{M}}

\newcommand\sC{\mathcal{C}}

\newcommand\sB{\mathcal{B}}

\newcommand*\id{\mathop{}\!\mathrm{d}}
\newcommand\RIPcst{\gamma}

\newcommand{\tam}{\mathrm{argmin}}

\newcommand{\tim}{\mathrm{Im}}

\newcommand{\ls}{\langle}
\newcommand{\rs}{\rangle}
\newcommand{\lsd}{{\boldsymbol \langle}}
\newcommand{\rsd}{{\boldsymbol \rangle}}
\newcommand{\re}{\mathcal{R}e}

\newcommand{\modif}[1]{\textcolor{black}{#1}}

\newcommand{\snewtext}{\color{black}}
\newcommand{\enewtext}{\color{black}}

\newcommand\indic{\mathbf{1}}
\newtheorem{theorem}{Theorem}[section]
\newtheorem{corollary}{Corollary}[section]
\newtheorem{proposition}{Proposition}[section]

\newtheorem{lemma}{Lemma}[section]
\newtheorem{definition}{Definition}[section]
\newtheorem{remark}{Remark}[section]
\newtheorem{assumption}{Assumption}[section]

\begin{document}


\title{The basins of attraction of the global minimizers of the non-convex sparse spike estimation problem}
\author{Yann Traonmilin$^{1,2,}$\footnote{Contact author : \url{yann.traonmilin@math.u-bordeaux.fr}}  and Jean-François Aujol$^{2}$ \\
$^1$CNRS,\\
$^2$Univ. Bordeaux, Bordeaux INP, CNRS,  IMB, UMR 5251,F-33400 Talence, France.\\}
\date{}
\maketitle

\begin{abstract}
The sparse spike estimation problem consists in estimating a number of  off-the-grid impulsive sources from under-determined linear measurements. Information theoretic results ensure that the minimization of  a non-convex functional is able to recover the spikes for adequately chosen measurements (deterministic or random). To solve this problem, methods inspired from the case of finite dimensional sparse estimation where a convex program is used have been proposed. Also greedy heuristics have shown nice practical results. However, little is known on the ideal non-convex minimization method. In this article, we study the shape of the global minimum of this non-convex functional: we give an explicit basin of attraction of the global minimum that shows that the non-convex problem becomes easier as the number of measurements grows. This has important consequences for methods involving descent algorithms (such as the greedy heuristic) and it gives insights for potential improvements of such descent methods.
\end{abstract}

\section{Introduction}

\subsection{Context}
Sums of sparse off-the-grid spikes can be used to model impulsive sources in signal processing (e.g. in astronomy, microscopy,...). Estimating such signals from a finite number of Fourier measurements is known as the super-resolution problem \cite{Candes_2014}. \modif{Also, the estimation of spikes from random Fourier measurements is at the core of the compressive $K$-means algorithm were $k$-means cluster centers are estimated from a compressed database \cite{Keriven_2017}.} In the space 
$\sM$ 
of finite signed measure over $\bR^d$, we aim  at recovering $x_0 = \sum_{i=1,k} a_i \delta_{t_i}$ from the measurements 
\begin{equation}
 y= Ax_0 + e,
\end{equation}
where $\delta_{t_i}$ is the Dirac measure at position $t_i$, the operator $A$ is a linear observation operator, $y \in \bC^m$ are the $m$ noisy measurements and $e$ is a finite energy observation noise.
Recent works have shown that it is possible to estimate spikes from a finite number of adequately chosen Fourier measurements as long as their locations are sufficiently separated,  using convex minimization based variational methods in the space of measures \cite{Candes_2013,Bhaskar_2013,Tang_2013,Castro_2015,Duval_2015}. Other general studies on inverse problems have shown that an ideal non-convex method (unfortunately computationally inefficient) can be used to recover these signals as long as the linear measurement operator has a restricted isometry property (RIP) \cite{Bourrier_2014}. In the case of super-resolution, adequately chosen random compressive measurements have been shown to meet the sufficient RIP conditions for separated spikes, thus guaranteeing the success of the ideal non-convex decoder \cite{Gribonval_2017}. 
These RIP results are based on an adequate kernel metric on  $\sM$.
It must be noted that, according to the work of~\cite{Bourrier_2014}, the success of the convex decoders as described in \cite{Candes_2013} for regular Fourier sampling implies a (lower) restricted isometry property of $A$  with respect to such a kernel metric (and not with the natural total variation metric: in this case no RIP is possible with finite regular Fourier measurements, see e.g.~\cite{Boyer_2017}).  Greedy heuristics have also been proposed to approach the non-convex minimization problem and they have shown good practical utility \cite{Keriven_2016,Keriven_2017,Traonmilin_2017}. 

While giving theoretical recovery guarantees, the convex-based method is non-convex in the space of parameters (amplitudes and locations) due to a polynomial root finding step. Also, it is difficult to implement in dimensions larger than one  in practice \cite{Duval_2017}. Greedy heuristics based on orthogonal matching pursuit are implemented in higher dimension (they can practically be used up to $d=50$), but they still miss theoretical recovery guarantees \cite{Keriven_2016}.  It would be possible to overcome the limitations of such methods if it were possible to perform the ideal non-convex minimization: 

\begin{equation} \label{eq:minimization}
    x^* \in \underset{x \in \Sigma}{\tam} \|Ax-y\|_2 
\end{equation}
where $\Sigma$ is a low-dimensional set modeling the separation constraints on the $k$ Diracs. \modif{Theoretical recovery guarantees for this minimization have been given in \cite{Gribonval_2017}.} While simple in its formulation, properties of this minimization procedure have not yet been thoroughly studied. 

In this article, as a first important step towards the understanding of the non-convex sparse spike estimation problem~\eqref{eq:minimization}, we study its formulation in the parameter space (the space of amplitudes and locations of the Diracs).  We observe that a smooth non-convex optimization can be performed. 

We place ourselves in a context where the number of measurements, either deterministic or random, guarantees the success of the ideal non-convex decoder with respect to a kernel metric $\|\cdot\|_h$, i.e. when we can ensure that:
\begin{equation}\label{eq:perf_bound}
 \|x^*-x_0\|_h\leq C \|e\|_2,
\end{equation}
where $C$ is an absolute constant with respect to $e$ and $x_0 \in \Sigma_{k,\epsilon}$, the set $\Sigma_{k,\epsilon}$ is the set of sums of $k$ spikes separated by $\epsilon$ on a given bounded domain. Qualitatively, the kernel metric can be viewed as a measure of the energy at a given resolution set by a kernel $h$ (see Section~\ref{sec:kernel_dipole}). 

The bound~\eqref{eq:perf_bound} is guaranteed by a restricted isometry property of $A$ defined using such kernel metric~\cite{Gribonval_2017}. This RIP setting is verified in the deterministic (see Section~\ref{sec:kernel_dipole}) and random weighted Fourier measurement contexts~\cite{Gribonval_2017}. We link this RIP of measurement operators with the conditioning of the Hessian of the global minimum, and we give an explicit basin of attraction of the global minimum \modif{in the parameter space}.  This study has direct consequences for the theoretical study of greedy approaches. Indeed a basin of attraction permits to give recovery guarantees for the gradient descent  (the initialization must fall within the basin), which is a step in the iterations of the greedy approach. 

\subsection{Parametrization of the model set $\Sigma$}

Let $\Sigma \subset \sM$ be a model set (union of subspaces) and $x_0 \in \Sigma$. Let $f(x) =\|Ax-y\|_2 $.

\begin{definition}[Parametrization of $\Sigma$]
A parametrization of  $\Sigma$ is a function $\phi$ such that $\Sigma \subset \phi(\bR^d) = \{\phi(\theta) : \theta \in \bR^d \}$.
\end{definition}
\snewtext
\begin{definition}[Local minimum]
The point $\theta \in \bR^d$ is a local minimum of $g : \bR^d \to \bR$ if there is $\epsilon > 0 $ such that for any $\theta' \in \bR^d$ such that $\|\theta-\theta'\|_2 \leq \epsilon$, we have $ g(\theta) \leq g(\theta')$.
\end{definition}
\enewtext
In the following, we consider the model of $\epsilon$-separated Diracs with $\epsilon >0$:
\begin{equation}
\begin{split}
\Sigma = \Sigma_{k,\epsilon} := \{ \phi(\theta)= \sum_{r=1,k} a_r \delta_{t_r} : \; &\theta = (a, t_1,..,t_k) \in \bR^{k(d+1)}, a \in \bR^k, t_r \in \bR^d,\\
&\forall  r \neq l, \|t_r-t_l\|_2> \epsilon, t_r \in  \sB_2(R)\},\\
\end{split}
\end{equation}
where  

\begin{equation} \label{eqB2}
\sB_2(R) = \{t \in \bR^d : \|t\|_2 \leq R\}.
\end{equation}

 Note that, in this paper, the Dirac distributions could be supported on any compact set. We use $\sB_2(R)$ for the sake of simplicity. For $t_r \in \bR^d$, we write $t_r =(t_{r,j})_{j=1,d}$.

We consider the following parametrization of $\Sigma_{k,\epsilon}$:  $\sum_{i=1,k} a_i \delta_{t_i} = \phi(\theta)  $ with $\theta= (a_{1},.., a_{k}, t_{1},..,t_{k})$.  We define 
\begin{equation} 
  \Theta_{k,\epsilon}:= \phi^{-1}(\Sigma_{k,\epsilon}).
\end{equation}

We consider the problem
\begin{equation} \label{eq:minimization2}
    \theta^* \in \arg \min_{ \theta \in E}  g(\theta)  = \arg \min_{ \theta \in E}  \|A\phi(\theta)-y\|_2  .
\end{equation}
where $E= \bR^{k(d+1)}$ or $E= \Theta_{k,\epsilon}$ and  $g(\theta)=  f(\phi(\theta))$. 

Note that when $E = \Theta_{k,\epsilon}$, performing minimization~\eqref{eq:minimization2} allows to recover the minima of the ideal minimization~\eqref{eq:minimization}, yielding stable recovery guarantees under a RIP assumption. Hence we are particularly interested in this case. When $E = \bR^{k(d+1)}$, we speak about unconstrained minimization for minimization~\eqref{eq:minimization2}. 

The objective of this paper is to study the shape of the basin of attraction of the global minimum of~\eqref{eq:minimization2} when $E = \Theta_{k,\epsilon}$. 

\subsection{Basin of attraction and descent algorithms}\label{sec:basin_def}

In this work, we are interested in minimizing $g$ defined in \eqref{eq:minimization2}.
Since $g$ is a smooth function, a classical method to minimize $g$ is to consider a fixed step gradient descent.
The algorithm is the following. Consider an initial point $\theta_0 \in \bR^d$ and a step size $\tau >0$. We define by recursion the sequence $\theta_n$ by
\begin{equation} \label{def_thetan}
\theta_{n+1} = \theta_n - \tau \nabla g(\theta_n)
\end{equation}

\modif{Such algorithm is used as a refinement step in the greedy heuristic based on orthogonal matching pursuit~\cite[Algorithm 1, Step 5]{Keriven_2017} in the practical setting of compressive statistical learning.}

We now give the definition of basin of attraction that we will use in this paper.
\begin{definition}[Basin of attraction] \label{defbassin}
We say that  a set $\Lambda \subset \bR^d$ is a basin of attraction of $g$ if there exists $\theta^* \in \Lambda$  and $\tau>0$, such that if $\theta_0 \in \Lambda$ then  the sequence $\theta_n$ defined by \eqref{def_thetan} converges to $\theta^*$.
\end{definition}

This definition of basin of attraction is related to the following classical optimization result (see e.g. \cite{ciarlet1989introduction}):

\begin{proposition}
Assume $g$ to be a smooth coercive convex function, whose gradient is $L$ Lipschitz.
Let $\theta_0 \in \bR^d$. Then, if $\tau < \frac{1}{ L}$, there exists $\theta^* \in \bR^d$ such that the sequence $\theta_n$ defined by \eqref{def_thetan} converges to $\theta^*$.
\end{proposition}

An immediate consequence of the previous proposition is the following corollary.
\begin{corollary}\label{cor:convergence_gradient_descent}
Assume $g$ to be a smooth function.
Assume that $g$ has a minimizer $\theta^* \in \bR^d$.
Assume that there exists an open set $\Lambda \subset \bR^d$ such that $\theta^* \in \Lambda$
, $g$ is convex on $\Lambda$ with  $L$ Lipschitz gradient.
\modif{Assume also that the sequence $\theta_n$ generated by the descent algorithm remains in $\Lambda$.}
Then, if $\theta_0 \in \Lambda$ and $\tau < \frac{1}{ L}$, 
the sequence $\theta_n$ defined by \eqref{def_thetan} converges to $\theta^*$.
\end{corollary}

\begin{remark}
Assume that $g$ is in $\sC^2$. 
Let $\lambda_{\max}(t)$ the largest  eigenvalue of the Hessian matrix of $g(t)$.
Let $\Theta \subset \bR^d$ an open set.
If there exists $L>0$ such that for all $t$ in $\Theta$, $\lambda_{\max}(t) \leq L$, then $g$ has a $L$ Lipschitz gradient in  $\Theta$.
\end{remark}

\modif{It is not obvious that the unconstrained gradient descent defined in iterations~\eqref{def_thetan} and the corresponding notion of basin of attraction is suitable to perform constrained minimization~\eqref{eq:minimization2}. In fact, we  show in this paper (essentially through Lemma~\ref{lem:link_unconstrained_constrained})  that the global minimum of constrained minimization~\eqref{eq:minimization2} has a basin of attraction.}

\subsection{Related work}

While original for the sparse spike estimation problem, it must be noted that the study of non-convex optimization schemes for linear inverse problems has gained attraction recently for different kinds of low-dimensional models. For low-rank matrix estimation, a smooth parametrization of the problem is possible and it has been shown that a RIP  guarantees the absence of spurious minima~\cite{Zhao_2015,Bhojanapalli_2016}. 
In~\cite{Waldspurger_2018}, a model for phase recovery  with alternated projections and smart initialization is considered. Conditions on the number of measurements guarantee the success of the technique. In the area of blind deconvolution and bi-convex programming, recent works have exploited similar ideas~\cite{Ling_2017,Cambareri_2018}. 

In the case of super-resolution, the idea of gradient descent has been studied in an asymptotic regime ($k\to \infty$) in~\cite{Chizat_2018} with theoretical conditions based on Wasserstein gradient flow for the initialization. In our case, we study the particular super-resolution problem with a fixed number of impulsions and we place ourselves in conditions where stable recovery is guaranteed, leading to explicit conditions on the initialization.

The objective of this article is to investigate to what extent these ideas can be applied to the theoretical study of the case of spike super-resolution estimation. 

The question of projected gradient descent raised in the last Section has been explored for general low-dimensional models \cite{Blumensath_2011}. It has been shown that the RIP guarantees the convergence of such algorithms with an ideal (often non practical) projection. Approached projected gradient descents have also been studied and shown to be successful for some particular applications~\cite{Golbabaee_2018}. The spikes super-resolution problem adds the parametrization step to these problems.

\subsection{Contributions and organization of the paper}
After a precise description of the setting, the definition of the kernel metric of interest and the associated restricted isometry for the spike estimation problem at the beginning of Section~\ref{sec:Hessian}, this article gives the following  original results:
\begin{enumerate}
 \item A bound on the conditioning of the Hessian  at a global minimum of the minimization in the parameter space is given in Section~\ref{sec:Hessian}. This bound shows that the better RIP constants are (RIP constants improve with respect to the number of measurements), the better the non-convex minimization problem behaves. It also shows that there is a basin of attraction of the global optimum  where no separation constraints are needed (for descent algorithms with an initialization close to the minimum, separation constraints can be discarded).
 \item An explicit shape of the basin of attraction of global minima is given in Section~\ref{sec:basin}. The size of the basin of attraction increases when the RIP constant gets better. 
\end{enumerate}
To conclude, we discuss the role of the separation constraint in descent algorithms in Section~\ref{sec:projected_gradient}, and we explain why enforcing a separation might improve them.  

\section{Conditioning of the Hessian} \label{sec:Hessian}

This section is devoted to the study of the Hessian matrix of $g$. In particular, we provide a bound on the conditioning of the Hessian  at a global minimum of the minimization in the parameter space. 

\subsection{Notations}
   The operator $A$ is a linear operator modeling $m$ measurements in $\bC^m$ ( $\tim A \subset \bC^m$ ) on the space of measures on $\bR^d$ defined by:  for $l=1,m$, 
\begin{equation} \label{eq:distribution}
(Au)_l =  \int_{\bR^d} \alpha_l(t) \id u(t)
\end{equation}
where $(\alpha_l)_l$ is a collection of functions in  $\sC^2(\sB_2(R))$ (twice continuously differentiable functions on $\sB_2(R)$ defined in \eqref{eqB2}).  

Notice that the integral used in \eqref{eq:distribution} is in fact a duality product \modif{$\lsd u, \alpha_l \rsd$} between a function in $\sC^2(\sB_2(R))$ and a finite signed measure over $\bR^d$. As the $\alpha_l$ are in $\sC^2(\sB_2(R))$, \modif{we can similarly apply $A$ to distributions of order 1 and 2 with support included in the relative interior of $\sB_2(R)$ which we note  $\text{rint} \sB_2(R)$.}

While a lot of  results for spike super-resolution are expressed on the $d$-dimensional Torus $\bT^d$, we prefer the setting of Diracs with bounded support on $\bR^d$ which is often closer to the physics of the considered phenomenom. However, our work is directly extended to the Torus setting by replacing $\bR^d$ by $\bT^d$ and $\sB^2(R)$ by $\bT^d$.

In $\bC^m$, we consider the Hermitian product $\ls x,y\rs = \sum x_i \bar{y}_i$. An example of such measurement operator is the (weighted) Fourier sampling: $(Au)_l = \frac{1}{\sqrt{m}} \int_{\bR^d}  c_l e^{-j \ls \omega_l,t\rs } \id u(t)$ for some chosen  frequencies $\omega_l \in \bR^d$ and frequency dependent weights $c_l \in \bR$.

Let $x = \sum_{i=1,k} a_i \delta_{t_i}$. By linearity of $A$, we have 
\begin{equation}
 (Ax)_l = \sum_{i=1}^k (A\delta_{t_i})_l =\sum_{i=1}^k a_i \alpha_l(t_i).
\end{equation}

With $g(\theta)=f(\phi(\theta))=  \|A \phi(\theta)-y\|_2^2$, we get:
\begin{equation}
g(\theta)=\sum_{l=1}^m \left|\sum_{i=1}^k a_i\alpha_l(t_i)-y_l\right|^2.
 \end{equation}

In the following, the notion of directional derivative will be important.
\begin{definition}[Directional derivatives]

Let $f$ be a $\sC^1$ function, and  $v \in \bR^d$ such that $\|v\|_2 = 1$. 
Then we can define the directional derivative of $f$ in direction $v$ by:
\begin{equation}
f_v'(t):=\langle v, \nabla f(t) \rangle=\lim_{h \to 0^+} \frac{f(t+hv)-f(t)}{h}
\end{equation}
Let $f$ be  a $\sC^2$ function, and  $(v_1,v_2) \in \bR^{2d}$ such that $\|v_1\|_2 = \|v_2\|_2 = 1$. Then we can define the second order directional derivative of $f$ in directions $v_1$ and $v_2$ by:
\begin{equation}
f_{v_1,v_2}''(t):=\langle v_1, \nabla^2 f(t) v_2 \rangle
\end{equation}
Notice that of course $f_{v_1,v_2}''(t)=f_{v_2,v_1}''(t)$. If $v_1=v_2$, we write $f_{v_1}''(t):=f_{v_1,v_1}''(t)$
\end{definition}

In particular, they permit to introduce derivatives of Dirac measures supported on $\bR^d$.

\begin{definition}[Directional derivatives of Dirac]

 Let $v \in \bR^d$ such that $\|v\|_2 = 1$. The distribution $\delta_{t_0,v}^\prime $ is  defined by \modif{$\lsd \delta_{t_0,v}^\prime , f\rsd= -f_v'(t_0) $}.  It is the limit of $\nu_\eta=-\frac{\delta_{t_0+\eta v} -\delta_{t_0}}{\eta}$ for $\eta \to 0^+$ in the distributional sense : for all $h \in \sC^1
 (\bR^d)$, $\int_{\bR}  h(t) \id\nu_\eta(t) \to_{\eta\to 0^+} \lsd \delta_{t_0,v}^\prime , h\rsd$. 
 
 Similarly, the distribution $\delta_{t_0,v}^{\prime \prime}$ is  defined by \modif{$\lsd  \delta_{t_0,v}^{\prime \prime}, f\rsd= f_v''(t_0)$} for $f \in  \sC^2 (\bR^d)$ and  the distribution $\delta_{t_0,v_1,v_2}^{\prime \prime}$ is  defined by \modif{$\lsd   \delta_{t_0,v_1,v_2}^{\prime \prime}, f\rsd= f_{v_1,v_2}''(t_0)$} for $f \in  \sC^2 (\bR^d)$ where $f_{v_1,v_2}''$ is the derivative of $f$ in direction $v_1$ chained with the derivative of $f$ in direction $v_2$.
 
 When $v = e_i$ is a vector of the canonical basis of $\bR^d$ , we write $\delta_{t_0,i}^{\prime}=\delta_{t_0,e_i}^{ \prime}$ and $\delta_{t_0,i}^{\prime \prime}= \delta_{t_0,e_i,e_i}^{\prime \prime}$.
\end{definition}

We now have the necessary tools to start the study of the Hessian of $g$.

\subsection{Gradient and Hessian of the objective function  $g$}\label{sec:gradient_Hessian}

We calculate the gradient and Hessian of $g$  in the two following propositions. We start with the gradient of $g$.
\begin{proposition}\label{prop:gradient}
For any $\theta \in \bR^{2k}$, we have:
\begin{equation}
 \begin{split}
  \frac{\partial g(\theta)}{\partial a_r}&= 2 \re \ls A\delta_{t_r}, A \phi(\theta)-y\rs, \\
 \end{split}
\end{equation}

\begin{equation}
 \begin{split}
 \frac{\partial g(\theta)}{\partial t_{r,j}}&= - 2 a_r\re \ls A \delta_{t_r,j}^{\prime}, A \phi(\theta)-y\rs.\\
 \end{split}
\end{equation}
\end{proposition}

\begin{proof}
See Appendix~\ref{proof21}.
\end{proof}

The next proposition gives the values of the Hessian matrix of $g$ which has a simple expression with the use of derivatives of Diracs. 

\begin{proposition}\label{prop:Hessian}
For any $\theta \in \bR^{k(d+1)}$
\begin{equation}
 \begin{split}
  H_{1,r,s}=\frac{\partial^2 g(\theta)}{\partial a_r \partial a_s}&= 2 \re \ls A\delta_{t_r}, A \delta_{t_s} \rs. \\
 \end{split}
\end{equation}
\begin{equation}
 \begin{split}
 H_{2,r,j_1,s,j_2}=\frac{\partial^2 g(\theta)}{\partial t_{r,j_1} \partial t_{s,j_2}}&=  2 a_r a_s \re \ls A \delta_{t_r,j_1}^{\prime}, A  \delta_{t_s,j_2}^{\prime}\rs 
 \\ &
 + \indic(r=s)2a_r\re \ls A \delta_{t_r,j_1,j_2}^{\prime \prime}, A\phi(\theta)-y \rs .\\
 \end{split}
\end{equation}

\begin{equation}
 \begin{split}
 H_{12,r,s,j} =\frac{\partial^2 g(\theta)}{\partial a_r \partial t_{s,j}}&=  -2  a_s \re \ls A \delta_{t_r}, A  \delta_{t_s,j}^{\prime}\rs -\indic(r=s) 2\re \ls A \delta_{t_s,j}^{\prime}, A\phi(\theta)-y \rs. \\
 \end{split}
\end{equation}
Hence the Hessian can be decomposed as the sum of two matrices $H= G + F$ with 
\begin{equation}
 \begin{split}
  G_{1,r,s}&= 2 \re \ls A\delta_{t_r}, A \delta_{t_s} \rs, \\
 G_{2,r,j_1,s,j_2}&=  2 a_r a_s \re \ls A \delta_{t_r,j_1}^{\prime}, A  \delta_{t_s,j_2}^{\prime}\rs ,
 \\
 G_{12,r,s,j}&= - 2  a_s \re \ls A \delta_{t_r}, A  \delta_{t_s,j}^{\prime}\rs.\\
 \end{split}
\end{equation}
and 
\begin{equation}
 \begin{split}
  F_{1,r,s}&= 0,\\
 F_{2,r,j_1,s,j_2}&=  \indic(r=s)2a_r\re \ls A \delta_{t_r,j_1,j_2}^{\prime \prime}, A\phi(\theta)-y \rs ,\\
 F_{12,r,s,j}&=  -\indic(r=s) 2\re \ls A \delta_{t_s,j}^{\prime}, A\phi(\theta)-y \rs.\\
 \end{split}
\end{equation}
\end{proposition}

\begin{proof}
See Appendix~\ref{proof21}.
\end{proof}

\subsection{Kernel, dipoles and the RIP }\label{sec:kernel_dipole}

In order to be able to build an operator $A$ with a RIP, we define a reproducible kernel Hilbert space (RKHS)  structure on the space of measures as in \cite{Gribonval_2017}, see also \cite{Sriperumbudur_2010}. The natural metric on the space of finite signed measures, the total variation of measures, is not well suited for a RIP analysis of the spikes super-resolution problems, as it does not measure the spacing between Diracs. When using the RIP, fundamental objects appear in the calculations: dipoles of Diracs. In this section we show that the typical RIP implies a RIP on dipoles and their generalization. 

\begin{definition}[Kernel, scalar product and norm]
 For finite signed measures over $\bR^d$, the Hilbert structure induced by a kernel $h$ (a smooth function from $\bR^d \times \bR^d \to \bR$) is  defined by the following scalar product between 2 measures $\pi_1,\pi_2$
\begin{equation}
 \ls \pi_1,\pi_2 \rs_h  = \int_{\bR^d}\int_{\bR^d} h(t_1,t_2) \id\pi_1(t_1)\id\pi_2(t_2).
\end{equation}
We can consequently define 
\begin{equation}
 \| \pi_1 \|_h^2 = \ls \pi_1,\pi_1 \rs_h.
\end{equation}
We have the relation 
\begin{equation}
 \| \pi_1 +\pi_2\|_h^2 = \| \pi_1 \|_h^2 +2\ls \pi_1,\pi_2 \rs_h +\| \pi_2\|_h^2.
\end{equation}
\end{definition}
Measuring distances with the help of $\|\cdot\|_h$ can be viewed as measuring distances at a given resolution set by $h$. Typically we use Gaussian kernels where the sharper the kernel is, the more accurate it is.

The next definition is taken from \cite{Gribonval_2017}.
\begin{definition}[($\epsilon$-)Dipole, separation]
 An $\epsilon$-dipole (noted dipole for simplicity) is  a measure $\pi = a_1\delta_{t_1}-a_2 \delta_{t_2}$ where $\|t_1-t_2\|_2 \leq \epsilon$. Two dipoles  $\pi_1 = a_1\delta_{t_1}-a_2 \delta_{t_2}$ and  $\pi_2 = a_3\delta_{t_3}-a_4 \delta_{t_4}$ are $\epsilon$-separated if their support are strictly $\epsilon$-separated (with respect to the $\ell^2$-norm on $\bR^d$), i.e. if  $\|t_1-t_3\|_2 > \epsilon$, $\|t_2-t_3\|_2 > \epsilon$ and  $\|t_1-t_4\|_2 > \epsilon$ and   $\|t_2-t_4\|_2 > \epsilon$.
\end{definition}

Compared to \cite{Gribonval_2017}, we need to introduce a new definition.
\begin{definition}[Generalized dipole] 
 A generalized dipole $\nu$ is either a dipole or  a distribution of order 1 of the form $a_1\delta_{t}+ a_2 \delta_{t,v}^{\prime}$. Two generalized dipoles are $\epsilon$-separated if their support are strictly $\epsilon$-separated (with respect to the $\ell^2$-norm on $\bR^d$).
\end{definition}
In this article we use regular, symmetrical, translation invariant kernels. Most recent developments to non translation invariant kernels~\cite{Poon_2018} could be considered to generalize this work, but they  are out of the scope of this article for the sake of simplicity.

\begin{assumption}\label{assum:kernel_prop}
A kernel $h$ follows this assumption if 
\begin{itemize}
 \item $h \in \sC^2(\bR^d,\bR^d)$.
 \item $h$ is symmetrical with respect to $0$, translation invariant, i.e. we can write $h(t_1,t_2)= \rho(\|t_1-t_2\|_2)$ where $\rho \in \sC^2(\bR)$.
  \item \modif{$h(t,t)= \rho(0) = 1 = \max_{t\in \bR^d,s\in\bR^d} |h(t,s)|$, $\rho'(0)= 0$, and $\rho''(0) < 0$}.
    \item \modif{there is a constant $c_h$ such that $0< c_h \leq \frac{\epsilon}{2}$ and  $\rho (t) \leq 1-\frac{|\rho''(0)|}{2}t^2$ for $t \in [0, c_h]$ (the existence of $c_h$ is a consequence of previous assumptions).}
    \item there is a constant $\mu_h$ such that, for all two $\epsilon$-separated dipoles, $\ls \nu_1,\nu_2 \rs_h \leq \mu_h \|\nu_1\|_h  \|\nu_2\|_h$ (mutual coherence).
\end{itemize} 
\end{assumption}

Note that the assumption that $h \in \sC^2$ guarantees the existence of integrals with respect to finite signed measures and duality product with distribution of order 1 with bounded supports.

\paragraph{Example} The now almost canonical well behaved kernel is the Gaussian kernel. From \cite{Gribonval_2017}, for $\epsilon=1$, using $h_0(t,s)= e^{-(t-s)^2/(2\sigma_k^2)} $ with $\sigma_k^2 = \frac{1}{2.4log(2k-1) +24}$, we have that $h_0$ follows Assumption~\ref{assum:kernel_prop} with  $\mu_{h_0} =\frac{3}{4(k-1)}$.
\\
The following Lemma and definition extend the scalar product induced by $h$ to generalized dipoles. 
\snewtext
\begin{lemma}\label{def:scalar_dip}
 Let $\nu_1 = a_1 \delta_{t_1} +b_1 \delta_{v_1,t_1}^{\prime}, \nu_2 = a_2 \delta_{t_2} +b_2 \delta_{v_2,t_2}^{\prime} $ be two generalized dipoles. Then  $\nu_1$ and $\nu_2$  are limits (in the distributional sense) of two sequences  of dipoles $\nu_1^{\eta_1} $ and $\nu_2^{\eta_2}$  for $\eta_1,\eta_2 \to 0$,   the quantity $\ls\nu_1^{\eta_1},\nu_2^{\eta_2} \rs_h $ converges, the limit is unique  (does not depend on the choice of  $\nu_1^{\eta_1} $ and $\nu_2^{\eta_2}$) and  

\begin{equation}
\begin{split}
   \lim_{\eta_1,\eta_2 \to 0} \ls\nu_1^{\eta_1},\nu_2^{\eta2}\rs_h 
  =& a_1 a_2f( t_1-t_2 )  - a_2 b_1  f_{v_1}^{\prime}(t_1-t_2) -a_1b_2f_{v_2}^{\prime}(t_2-t_1)\\
  &-b_1b_2  f_{v_1,v_2}^{\prime \prime}(t_1- t_2  ) \\
  \end{split}
\end{equation}
where $f(t) = \rho(\|t\|_2)$.
\end{lemma}

\begin{proof}
See Appendix~\ref{proof22}.
\end{proof}

\begin{definition}
 Let $\nu_1 = a_1 \delta_{t_1} +b_1 \delta_{v_1,t_1}^{\prime}, \nu_2 = a_2 \delta_{t_2} +b_2 \delta_{v_2,t_2}^{\prime} $ be two generalized dipoles. With the previous Lemma,  we define 

\begin{equation}
\begin{split}
  \ls\nu_1,\nu_2 \rs_h &:= \lim_{\eta_1,\eta_2 \to 0} \ls\nu_1^{\eta_1},\nu_2^{\eta2}\rs_h \\
  \end{split}
\end{equation}
where    $\nu_1^{\eta_1} $ and $\nu_2^{\eta_2}$  are  two sequences  of dipoles  that converge to  $\nu_1$ and $\nu_2$   (in the distributional sense) for $\eta_1,\eta_2 \to 0$.
\end{definition}

\enewtext

We have the following properties that are immediate consequences of Lemma~\ref{def:scalar_dip}.   
\begin{lemma}\label{lem:kernel_dirac_properties}
 Let $h$ be a kernel meeting Assumption~\ref{assum:kernel_prop}.  We have the following properties for any $t \in \bR$:
 \begin{equation}
 \|\delta_t\|_h^2 = \rho(0) = 1
 \end{equation}
\begin{equation}
  \ls \delta_t, \delta_{t,v}^{\prime}\rs_h = -\rho'(0) = 0
 \end{equation}
\begin{equation}
  \|\delta_{t,v}^{\prime}\|_h^2  =|\rho''(0)|
 \end{equation}
\end{lemma}

\begin{proof}
See Appendix~\ref{proof22}.
\end{proof}

From \cite[Lemma 6.5]{Gribonval_2017}, we have the following Lemma:
\begin{lemma}\label{lem:pyth_dipole}
  Suppose for all two $\epsilon$-separated dipoles, $\ls \pi_1,\pi_2 \rs_h \leq \mu \|\pi_1\|_h  \|\pi_2\|_h$ (mutual coherence). Then for $k$, $\epsilon$-separated dipoles $\pi_1, ...\pi_k$ such that $\max_i \|\pi_i\|_h>0$, we have 
  \begin{equation}
   1 - (k-1)\mu \leq \frac{\|\sum_{i=1,k} \pi_i\|_h^2}{\sum_{i=1,k}\| \pi_i\|_h^2} \leq  1+ (k-1) \mu.
  \end{equation}
\end{lemma}

We can generalize the previous result to generalized dipoles.
\begin{lemma}\label{lem:dip2Generalizeddip}
 Let two $\epsilon$-separated \textbf{generalized} dipoles $\nu_1,\nu_2$. Suppose for all two $\epsilon$-separated dipoles $\pi_1,\pi_2$, $\ls \pi_1,\pi_2 \rs_h \leq \mu \|\pi_1\|_h  \|\pi_2\|_h$ (mutual coherence). Then we have:
\begin{equation}
 \ls\nu_1, \nu_2\rs_h  \leq  \mu \|\nu_1\|_h  \|\nu_2\|_h
\end{equation}
\end{lemma}

\begin{proof}
See Appendix~\ref{proof22}.
\end{proof}

A consequence of the previous result is the following Lemma:
\begin{lemma} \label{lem:mutual_Generalized_dipoles}
  Suppose for all two $\epsilon$-separated generalized dipoles, $\ls \nu_1,\nu_2 \rs_h \leq \mu \|\nu_1\|_h  \|\nu_2\|_h$ (mutual coherence). Then for $k$ $\epsilon$-separated generalized dipoles $\nu_1, ...\nu_k$ such that
  $\max_i \|\nu_i\|_h>0$, we have 
  \begin{equation}
   1 - (k-1)\mu \leq \frac{\|\sum_{i=1,k} \nu_i\|_h^2}{\sum_{i=1,k}\| \nu_i\|_h^2} \leq  1+ (k-1) \mu.
  \end{equation}
\end{lemma}

\begin{proof}
See Appendix~\ref{proof22}.
\end{proof}

We are now able to define the Restricted Isometry Property (RIP). The secant set of the model set $\Sigma$ is $\Sigma - \Sigma := \{ x-y : x \in \Sigma, y \in \Sigma\}$.
\begin{definition}[RIP]
$A$ has the RIP on $\Sigma-\Sigma$ with respect to $\|\cdot\|$  with constant $\RIPcst$ if for all $x \in \Sigma -\Sigma$:
\begin{equation}\label{eq:DefRIP}
(1-\RIPcst)\|x\|^2 \leq \|A x\|_2^2 \leq (1+\RIPcst)\|Ax\|^2.
\end{equation}
\end{definition}

In the following we will suppose that $A$ has RIP $\RIPcst$ on $\Sigma_{k,\epsilon}-\Sigma_{k,\epsilon}$ with respect to $\|\cdot\|_h$, i.e. for $\sum_{r=1,k}a_r\delta_{t_r}-\sum_{r=1,k} b_r \delta_{s_r} \in \Sigma_{k,\epsilon} -\Sigma_{k,\epsilon}$, we have 
\begin{eqnarray}
  (1-\RIPcst)\left\|\sum_{r=1,k} (a_r\delta_{t_r}-b_r\delta_{s_r} )\right\|_h^2 
  & \leq & \left\|A\sum_{r=1,k} (a_r\delta_{t_r}- b_r\delta_{s_r})\right\|_2^2   
  \\
   & \leq & (1+\RIPcst) \left\|\sum_{r=1,k}a_r\delta_{t_r}-b_r\delta_{s_r}\right\|_h^2.  \nonumber
\end{eqnarray}

From~\cite{Gribonval_2017}, with a Gaussian kernel $h$   it is possible to build a random $A$ with RIP constant $\RIPcst$. With this choice of $A$, the ideal minimization~\eqref{eq:minimization} yields a stable and robust estimation of $x_0$ with respect to the $\|\cdot\|_h$. 

In \cite{Candes_2013}, stable recovery for $\epsilon$-separated Diracs  is guaranteed on the Torus with the metric $\|K_{hi}*\cdot\|_{L^1}$ where $K_{hi}*$ is the convolution with a Fejér kernel. From \cite[IV.A]{Bourrier_2014}, this guarantees a lower RIP with respect to this metric. \modif{Indeed, the $L^1$-norm of trigonometric polynomials (on $[0,1]$) is lower bounded by their $L^2$-norm, i.e. there is an absolute constant $D > 0$ depending on $K_{hi}$ such that $\|K_{hi}*\cdot\|_{L^1} \geq D\|K_{hi}*\cdot\|_{L^2}$ (see \cite[p. 230]{Timan_1963}).  Applying Lemma~\ref{lem:rkhs_convol} from the Annex on the Fejér kernel shows that there exists a kernel metric $\|\cdot\|_{h_K}$  that lower bounds $\|K_{hi}*\cdot\|_{L^1}$ for sums of Diracs.} This guarantees the existence of a lower RIP with respect to a kernel metric for the conventional deterministic spike super-resolution setting.

The RIP on $\Sigma_{k,\epsilon}-\Sigma_{k,\epsilon}$  implies a RIP on $\epsilon$-separated generalized dipoles.

\begin{lemma}[RIP on generalized dipoles]\label{lem:RIP_derivative}
 Suppose $A$ has the RIP on  $\Sigma_{k,\epsilon}-\Sigma_{k,\epsilon}$ with constant $\RIPcst$. Let $(\nu_r)_{r=1,k}$, $k$ $\epsilon$-separated dipoles supported in $\text{rint} \sB_2(R)$, we have 
\begin{equation}
\begin{split}
 (1-\RIPcst)\left\| \sum_{r=1,k}\nu_r\right\|_h^2 \leq \left\|A(\sum_{r=1,k} \nu_r)\right\|_2^2   \leq (1+\RIPcst)\left\|\sum_{r=1,k} \nu_r \right\|_h^2.  \\
\end{split}
 \end{equation}
\end{lemma}

\begin{proof}
See Appendix~\ref{proof22}.
\end{proof}

Finally, we will need a last estimate. To state it, we need first to introduce the following definition:
\begin{definition}
Let $A$ such that the $\alpha_l$ are in $\sC^2(\sB_2(R))$. We define
\begin{equation} \label{defDAR}
D_{A,R} :=\sup_{1\leq l \leq m ; v\in \bR^d,w\in\bR^d: \|v\|_2= \|w\|_2=1;  t \in \sB_2(R)} | \alpha_{l,v,w}^{\prime \prime}(t)| .
\end{equation}
The constant $D_{A,R}$ is finite, and it is thus a  bound of the directional second derivatives of the $\alpha_l$ over $\sB_2(R)$.

\end{definition}

\begin{lemma} \label{lem:upper_RIP_second_deriv}
Let $A$ such that the $\alpha_l$ are in $\sC^2(\sB_2(R))$. Then, for any $t \in \sB_2(R)$, with directions $v_1,v_2$, we have 
 \begin{equation}
   \|A\delta_{t,v_1,v_2}^{\prime \prime}\|_2 \leq   \sqrt{m}D_{A,R}.
 \end{equation}
 where $D_{A,R}$ is defined in Equation~\eqref{defDAR}.
\end{lemma}
\begin{proof}
See Appendix~\ref{proof22}.
\end{proof}

\subsection{Control of the conditioning of the Hessian with the restricted isometry property}\label{sec:control_Hessian}

We can now give  a lower (resp. upper) bound for the highest (resp. lowest) eigenvalues of the Hessian matrix $H$ of $g$ (computed in Proposition~\ref{prop:Hessian}).

\begin{theorem}[Control of the Hessian]\label{th:min_max_eigen_control_H}
 Let $\theta = (a_1,..,a_k, t_1,..t_k) \in \Theta_{k,\epsilon}$ with $t \in \text{rint}\sB_2(R)$ and  $\theta^* \in \Theta_{k,\epsilon}$ a minimizer of~\eqref{eq:minimization2}. Suppose $h$ follows Assumption~\ref{assum:kernel_prop}. Let $H$ the Hessian of $g$ at $\theta$. Suppose $A$ has RIP $\RIPcst$ on $\Sigma_{k,\epsilon}-\Sigma_{k,\epsilon}$. We have 
 \begin{equation} \label{largeeigenvalue}
 \sup_{\|u\|_2 =1} u^THu \leq  2(1+\RIPcst)(1+(k-1)\mu)\max(1,(a_r^2|\rho''(0)|)_{r=1,l}) + \xi; \\
\end{equation}
\begin{equation}\label{smalleigenvalue}
 \inf_{\|u\|_2 =1} u^THu \geq 2(1-\RIPcst)(1-(k-1)\mu)\min(1,(a_r^2|\rho''(0)|)_{r=1,l})-\xi\\
\end{equation}
where $\xi = 2(d+1)\max( \max_{r}|a_r|  \sqrt{m}D_{A,R} ,\sqrt{1+\RIPcst}\sqrt{|\rho''(0)|} ) (\|A\phi(\theta)-A\phi(\theta^*)\|_2+ \|e\|_2)$, the constant $D_{A,R}$ is defined in \eqref{defDAR}
 and $e$ is the finite energy measurement noise.
\end{theorem}

\begin{proof}
See Appendix~\ref{proof23}.
\end{proof}

\begin{remark}
 Notice that, in the noiseless case, \eqref{smalleigenvalue} ensures in particular that $g$ has a positive Hessian matrix in  $\theta^* $. Moreover, if $\min_r |a_r| >0$, there exists a neighbourhood of $\theta^* $, in which $g$ remains convex. We will give an explicit size for this neighbourhood in the next section. 
 Notice also that \eqref{largeeigenvalue} gives an upper bound on the Lipschitz constant of the gradient of $g$. This implies the existence of a basin of  attraction (see Definition~\ref{defbassin}) with a uniform bound for the step size.
 \end{remark}

\begin{remark}
 With the method to choose $A$ from  \cite[Lemma 6.5]{Gribonval_2017}, for any $\RIPcst$ and $m \gtrsim k^2d \text{polylog}(k,d)/\RIPcst^2$, we can find $A$ that has RIP with high probability with a kernel $h_0$ having the right properties.
 \end{remark}

We can  control the conditioning of the Hessian matrix $\kappa(H)$ at a global minimum as the term $\|A\phi(\theta)-A\phi(\theta^*)\|_2$ vanishes in the control from Theorem~\ref{th:min_max_eigen_control_H}. Particularly, in the noiseless case we have the following Corollary.  The lower bound is useful to confirm the dependency on the ratio of amplitudes when it converges to $+\infty$. For this next result, we  make the additional assumption that $\min_r |a_r| > 0$. In practice, this amounts to assuming that when estimating the Diracs, we do not over-estimate their number (which will often be the case, in particular in the presence of noise). When the number of Diracs is overestimated, the minimizers of~\eqref{eq:minimization2} are points that are not isolated, the notion of basin of attraction would have to be generalized to a basin of attraction of a set of minimizers (when $a_r = 0$, $g(\theta)$ does not depend on $t_r$), which is out of the scope of this article for clarity purpose.

\begin{corollary}\label{cor:control_Hessian}
 Let $x_0 = \sum_{r=1,k} a_r \delta_{t_r} \in \Sigma_{k,\epsilon} = \phi(\theta_0)$ and $e=0$. Suppose $h$ follows Assumption~\ref{assum:kernel_prop}. Let $H$ the Hessian of $g$ at $\theta_0$. Suppose $A$ has RIP $\RIPcst$ on $\Sigma_{k,\epsilon}-\Sigma_{k,\epsilon}$, and  that $\min_r |a_r| > 0$.
 We have 
\begin{equation}
\begin{split}
  \frac{(1-\RIPcst)\max(1,(a_r^2|\rho''(0)|)_{r=1,l})}{(1+\RIPcst)\min(1,(a_r^2|\rho''(0)|)_{r=1,l})}  &  \leq \kappa(H) 
  \\ & \leq \frac{(1+\RIPcst)(1+(k-1)\mu)\max(1,(a_r^2|\rho''(0)|)_{r=1,l})}{(1-\RIPcst)(1-(k-1)\mu)\min(1,(a_r^2|\rho''(0)|)_{r=1,l})}.
  \end{split}
\end{equation}
\end{corollary}

\begin{proof}
See Appendix~\ref{proof23}.
\end{proof}

It is easy to see that for a noise $e$ with small enough energy (i.e. such that $\xi$ is strictly lower than $2(1-\RIPcst)(1-(k-1)\mu)\min(1,(a_r^2|\rho''(0)|)_{r=1,l})$, if  $\min_r |a_r| > 0$, then the Hessian at a global minimum is strictly positive. Of course, this may require a very small noise since the ratio of amplitudes at the global minimum can be large.  

\begin{remark}
 We remark that for a same maximal ratio of amplitudes in $\theta^*$, a better conditioning bound is achieved when $\max_{r=1,l}a_r^2 |\rho''(0)| \geq 1 \geq  \min_{r=1,l}a_r^2 |\rho''(0)|$.  We attribute this to the fact that we estimate amplitudes and locations at the same time. The amplitudes must be appropriately scaled to match the variations of $g$ with respect to locations. Intuitively, alternate descent with respect to amplitudes and locations might be better than the classical gradient descent for easily setting the descent step.
 \end{remark}

\begin{remark}
 As $g$ is $\sC^2$, ensuring the strict positivity of the Hessian at the global minimum guarantees the \emph{existence}  of a basin of attraction as emphasized in Section~\ref{sec:basin_def}. In the next Section, we give an explicit formulation of a basin of attraction.
\end{remark}

\section{Explicit basin of attraction of the global minimum} \label{sec:basin}

Let $\theta_1 \in \bR^d$. Can we guarantee, for some notion of distance $d$, that $d(\theta_1,\theta^*)\leq C$ and $\theta_1 \neq \theta^*$, with $C$ an explicit constant, implies $\nabla g(\theta_1) \neq 0$ ?  The following theorems show that it is in fact the case. With a strong RIP assumption, we can give an explicit basin of attraction of the global minimum for minimization~\eqref{eq:minimization2} without separation constraints.

\subsection{Uniform control of the Hessian}
In the noiseless case, a global minimum $\theta^*$ of the  constrained minimization of $g$ over $\Theta_{k,\epsilon}$ is also a global minimum of the unconstrained minimization because $g(\theta^*) = 0$. 
In the presence of noise, we can no longer guarantee that the minimizer of the constrained problem $\theta^*$ is a global minimum of the unconstrained problem. However, the shape of the constraint guarantees that it is a local minimum (see next Lemma).

\begin{lemma} \label{lem:link_unconstrained_constrained}
 Suppose $\theta^* =(a_1,..,a_k,t_1,..,t_k)$ is a result of constrained minimization~\eqref{eq:minimization2} with $t_i  \in \text{rint} \sB_2(R)$. Then  $\theta^*$ is a local minimum of $g$.
\end{lemma}
\begin{proof}
 let $\theta^*=(a_1,..,a_k,t_1,..,t_k)$.  As for all $i \neq j$,  $\|t_i-t_j\|_\infty> \epsilon$, there exists $\eta >0$ such that for all $\theta = (b_1,..,b_k,s_1,..,s_k)$  such that $\|s_i-t_i\|_\infty < \eta$, we have $\theta \in \Theta_{k,\epsilon}$. Hence, $\theta^* + B_\infty(\eta) \subset \Theta_{k, \epsilon}$, and $\theta^* \in \arg \min_{\theta \in \theta^* + B_\infty(\eta) } g(\theta)$.
\end{proof}

 Hence we can still calculate a basin of attraction of $\theta^*$ (for the unconstrained minimization).  The expression of the basin in the next Section is a direct consequence of the following Theorem that uniformly control the Hessian of $g$ in an explicit neighbourhood of $\theta^*$.

\begin{theorem}\label{th:basin}
 Suppose $A$ has RIP $\RIPcst$ on $\Sigma_{k,\frac{\epsilon}{2}}-\Sigma_{k,\frac{\epsilon}{2}}$ and that $h$ follows Assumption~\ref{assum:kernel_prop} and has mutual coherence constant $\mu$ on $\frac{\epsilon}{2}$-separated dipoles.
Let $\theta^*=(a_1,..,a_k,t_1,..,t_k) \in \Theta_{k,\epsilon}$ be a result of constrained minimization~\eqref{eq:minimization2} such that $t_i \in \text{rint} \sB_2(R)$. Suppose $0<|a_1| \leq |a_2| ... \leq |a_k|$. Let \modif{$0 \leq \beta \leq \frac{\epsilon}{4}$} and
\snewtext
\begin{equation}
\begin{split}
\Lambda_{\theta^*,\beta}  
:=  \{& \theta:  \|\theta-\theta^*\|_2 < \beta \}.
\end{split}
\end{equation}
\enewtext
If $\theta \in \Lambda_{\theta^*,\beta}$,
then $H$ the Hessian of $g$ at $\theta$ has  the following bounds : 
\begin{equation}
 \sup_{\|u\|_2 =1} u^THu \leq  2(1+\RIPcst)(1+(k-1)\mu)\max(1,\modif{(|a_k|+ \beta)}^2|\rho''(0)|) + \xi; \\
\end{equation}
\begin{equation}
 \inf_{\|u\|_2 =1} u^THu \geq 2(1-\RIPcst)(1-(k-1)\mu)\min(1,\modif{(|a_1| -\beta )}^2|\rho''(0)|)-\xi\\
\end{equation}
where $\xi =2(d+1)\max( |a_k|  \sqrt{m}D_{A,R} ,\sqrt{1+\RIPcst}\sqrt{|\rho''(0)|} ) (\sup_{\theta \in \Lambda_{\theta^*,\beta}} \|A\phi(\theta)-A\phi(\theta^*)\|_2+ \|e\|_2)$, the constant
$D_{A,R}$ is given in~\eqref{defDAR}
and $e$ is the finite energy measurement noise.
\end{theorem}

 \begin{proof}
See Appendix~\ref{proof3}.
\end{proof}

\begin{remark}
We observe that we require a stronger RIP than the usual one on  $\Sigma_{k,\epsilon} - \Sigma_{k,\epsilon}$ to guarantee that unconstrained minimization converges in the basin of attraction  $ \Lambda_{\theta^*,\beta} $.
\end{remark}

\modif{The set $\Lambda_{\theta^*,\beta} $ is an open $\ell^2$ ball centered on $\theta^*$. The choice of this set, besides its simplicity, is useful to guarantee the convergence of the gradient descent.  We could garantee the positivity of the Hessian on bigger sets with more complicated formulations. Guaranteeing that iterates of the gradient descent stay in such sets would become much harder then.}  When the separation constraint is added for the basin of attraction (we look for potential critical points in $\Sigma_{k,\epsilon}$), we can provide better bounds. We will discuss what we could expect from constrained descent algorithms in Section~\ref{sec:projected_gradient}.
 
\begin{theorem}\label{th:basin_with_constraint}
Suppose $A$ has RIP $\RIPcst$ on $\Sigma_{k,\epsilon}-\Sigma_{k,\epsilon}$ and that $h$ follows Assumption~\ref{assum:kernel_prop} and has mutual coherence constant $\mu$ on $\epsilon$-separated dipoles.
 Suppose $0<|a_1| \leq |a_2| ... \leq |a_k|$. Let $\theta^*=(a_1,...,a_k,t_1,..t_k)\in \Theta_{k,\epsilon}$ be a result of constrained minimization~\eqref{eq:minimization2} such that $t_i \in \text{rint} \sB_2(R)$. Let $ \beta \geq 0 $ and
\snewtext
\begin{equation}
\begin{split}
 \Lambda_{\theta^*,\beta}  
:=  \{& \theta:  \|\theta-\theta^*\|_2 < \beta \}.
\end{split}
\end{equation}
\enewtext
  Then for $\theta \in \Theta_{k,\epsilon} \cap \Lambda_{\theta^*,\beta} $, then $H$ the Hessian of $g$ at $\theta$ has  the following bounds: 
\begin{equation}
 \sup_{\|u\|_2 =1} u^THu \leq  2(1+\RIPcst)(1+(k-1)\mu)\max(1,\modif{(|a_k|+ \beta)}^2|\rho''(0)|) + \xi; \\
\end{equation}
\begin{equation}
 \inf_{\|u\|_2 =1} u^THu \geq 2(1-\RIPcst)(1-(k-1)\mu)\min(1,\modif{(|a_1| -\beta)}^2|\rho''(0)|)-\xi\\
\end{equation}
where $\xi =2(d+1)\max( |a_k|  \sqrt{m}D_{A,R} ,\sqrt{1+\RIPcst}\sqrt{|\rho''(0)|} ) (\sup_{\theta \in \Lambda_{\theta^*,\beta}} \|A\phi(\theta)-A\phi(\theta^*)\|_2+ \|e\|_2)$, 
 the constant $D_{A,R}$ is given in \eqref{defDAR}
and $e$ is the finite energy measurement noise.

 \end{theorem}
 
 \begin{proof}
See Appendix~\ref{proof3}.
\end{proof}

 \subsection{Explicit basin of attraction in the noiseless and noisy case }
 
With the help of this uniform control of the Hessian we give an explicit (yet suboptimal)  basin of attraction. 
  
\begin{corollary}[of Theorem~\ref{th:basin}, noiseless case]\label{cor:basin_noiseless}
Under the hypotheses of Theorem~\ref{th:basin}, let $\theta^* \in \Theta_{k,\epsilon}$ be a result of constrained minimization~\eqref{eq:minimization2}. Let $a^* =(a_1,a_2...,a_k)$.
Take  
\snewtext
$$\beta_{max}  := \min \left( c_h, \frac{|a_1|}{2}, C_1C_2 \right)$$ where $C_1 = \frac{ (1-\RIPcst)(1-(k-1)\mu) }{ (d+1)\sqrt{1+\RIPcst}  \sqrt{1+(k-1)\mu} } $ and $C_2 = \frac{ \min(1,|a_1|^2|\rho''(0)|/4) }{  \max( |a_k|  \sqrt{m}D_{A,R} ,\sqrt{1+\RIPcst}\sqrt{|\rho''(0)|} )  \sqrt{ 1 + 2|\rho''(0)| \|a^*\|_2^2}}$
\enewtext
 Then the set 
 $\Lambda_{\theta^*,\beta_{max}}$ 
 is a basin of attraction of $\theta^*$.
\end{corollary}
  
 \begin{proof}
See Appendix~\ref{proof3}.
\end{proof}

  The parameter $\beta$ controls the distance between a parameter and the optimal parameter. When the RIP constant $\gamma$ decreases (and generally as the number of measurement increases), the size of the basin of attraction increases.  In both the context of regular Fourier sampling and random Fourier sampling, the constant $D_{A,R}$ is bounded when $m$ increases. When the mutual coherence constant $\mu$ decreases, the basin of attraction also increases. \modif{The size of the basin also decreases as the ratio of amplitudes $\frac{a_1}{a_k}$ decreases. We observe again that performing the descent with respect to amplitudes and positions at the same time yields pessimistic bounds for the basin of attraction}. Finally, we note that the smaller $\beta$ is, the  smaller is the upper bound on the operator norm of the Hessian.

When the noise contaminating the measurements is small enough, we have similar results with a smaller basin of attraction.

\begin{corollary}[of Theorem~\ref{th:basin}, noisy case]\label{cor:basin_noisy}
Under the hypotheses of Theorem~\ref{th:basin}, let $\theta^* \in \Theta_{k,\epsilon}$ be a result of constrained minimization~\eqref{eq:minimization2}. . Let $a^* =(a_1,a_2...,a_k)$.
Take  
\snewtext
$$\beta_{max}  := \min \left( c_h, \frac{|a_1|}{2}, C_1C_3 \right)$$ where $C_1 = \frac{ (1-\RIPcst)(1-(k-1)\mu) }{ (d+1)\sqrt{1+\RIPcst}  \sqrt{1+(k-1)\mu} } $ and $C_3 = \frac{ \min(1,|a_1|^2|\rho''(0)|/4) }{  \max( |a_k|  \sqrt{m}D_{A,R} ,\sqrt{1+\RIPcst}\sqrt{|\rho''(0)|} ) (1+ \sqrt{ 1 + 2|\rho''(0)| \|a^*\|_2^2})}$

Suppose $\|e\|_2 \leq  \sqrt{1+\RIPcst} \sqrt{1+(k-1)\mu} \beta $. \enewtext
 Then the set 
 $\Lambda_{\theta^*,\beta_{max}}$
 is a basin of attraction of $\theta^*$.
\end{corollary}
  
 \begin{proof}
See Appendix~\ref{proof3}.
\end{proof}

 \section{Towards new descent algorithms for sparse spike estimation?} \label{sec:projected_gradient}

 We have shown that, given an appropriate measurement operator for separated Diracs, a good initialization is sufficient to guarantee the success of a simple gradient descent. Such gradient descent is used is  in the practical setting of compressive statistical learning \cite{Keriven_2017}. \modif{Our result on unconstrained minimization explains why the use of such gradient descent is valid in this setting}. If we could guarantee additionally that by greedily estimating Diracs, we fall within the basin of attraction, we would have a full non-convex optimization technique with guarantees of convergence to a global minimum.

In other works~\cite{Duval_2017thin1,Duval_2017thin2}, it has been shown that discretization (on grids) of convex methods have a tendency to produce spurious spikes at Dirac locations. Our results seem to indicate that merging spikes that are close to each other when performing a gradient descent might break the barrier between continuous and discrete methods. 

Theorem~\ref{th:basin_with_constraint} brings another question as the Hessian of $g$ is more easily controlled in $\Theta_{k,\epsilon}$.  More generally, can we build a simple descent algorithm that stays in $\Theta_{k,\epsilon}$ to get larger basins of attraction? Consider the problem for $d=1$ in the noiseless case for the sake of clarity. We want to use the following descent algorithm:
\begin{equation}
 \theta_{i+1} = P_{\Theta_{k,\epsilon}}(\theta_i - \tau \nabla g(\theta_i))
\end{equation}
Where $P_{\Theta_{k,\epsilon}}$ is a projection onto the separation constraint.  Notice that since $\Theta_{k,\epsilon}$ is not a convex set, we cannot easily define the orthogonal projection onto it (it may not even exists).

If we suppose that the gradient descent step decreases $g$ (i.e. $g(\theta_i - \tau \nabla g(\theta_i))< g(\theta_i)$), is it possible to guarantee that applying projection step keeps decreasing $g$? Consider: 
\begin{equation} \label{pseudoproj}
  P_{\Theta_{k,\epsilon}}(\theta) \in \arg \min_{\tilde{\theta} \in \Theta_{k,\epsilon}} \left|\|A \phi(\tilde{\theta})-y \|_2 -  \|A \phi(\theta)-y \|_2|\right|
\end{equation}

First consider the following Lemma: 
\begin{lemma}\label{lem:pseudo_conv}
 Let $d=1$. Let $\theta_0,\theta_1 \in \Theta_{k,\epsilon}$.  Let $g(\theta) = \|A\phi(\theta) - A \phi(\theta_0)\|$. Then for all $\alpha$ such that $0 = g(\theta_0) \leq \alpha \leq g(\theta_1)$, there exists $\theta^* \in \Theta_{k,\epsilon}$ such that $g(\theta^*) = \alpha$.
\end{lemma}

\begin{proof}
See Appendix~\ref{proof4}.
\end{proof}

Lemma~\ref{lem:pseudo_conv} essentially guarantees that is is possible to continuously map the interval $[0,g(\theta_1)]$ by $g$ with elements of $ \Theta_{k,\epsilon}$.
Hence, at a step $i+1$, we have 
\begin{equation}
 |g(\theta_{i+1}) -g(\theta_i)| = |g(\theta_i - \tau \nabla g(\theta_i)) -g(\theta_i)|.
\end{equation}

 The projection $P_{\Theta_{k,\epsilon}}$  defined by \eqref{pseudoproj} is not easy to calculate (in fact, it is a similar optimization problem as the main problem). Other more "natural" projections on $\Theta_{k,\epsilon}$  could be defined as : 
 \begin{equation}
  P_{\Theta_{k,\epsilon}}(\theta) \in \phi^{-1} (\arg \inf_{x \in \Sigma_{k,\epsilon}}\|A x-A \phi(\theta) \|_2)
\end{equation}
or 
 \begin{equation}
  P_{\Theta_{k,\epsilon}}(\theta) \in \phi^{-1}( \arg \inf_{x \in \Sigma_{k,\epsilon}} \|x- \phi(\theta) \|_h).
\end{equation}
However they suffer from the same calculability drawback. This suggests to build a new family of heuristic algorithms of spike estimation where we propose heuristics to approach the projection of $\hat{\theta}_{i+1}$ on $\Theta_{k,\epsilon}$. Recovery guarantees would be obtained by guaranteeing that the projection heuristic does not increase the value of $g$ by too much compared to the gradient descent step.

 
\appendix 
 
\section{Annex}
 \subsection{Proofs for Section~\ref{sec:gradient_Hessian}} \label{proof21}

\begin{proof}[Proof of Proposition~\ref{prop:gradient}]
 \begin{equation}
  \begin{split}
   \frac{\partial g(\theta)}{\partial a_r} &=  \frac{\partial }{\partial a_r} \sum_{l=1}^m \left|\sum_{i=1}^k a_i\alpha_l(t_i)-y_l\right|^2\\
 &=  \sum_{l=1}^m 2 \re  \left(\alpha_l(t_{r,j}) \left(\overline{\sum_{i=1}^k a_i\alpha_l(t_i)-y_l}\right)\right)   \\
 &= 2 \re \ls A\delta_{t_r}, A \phi(\theta)-y\rs.\\
 \end{split}
 \end{equation}
Similarly,
 \begin{equation}
  \begin{split}
   \frac{\partial  g(\theta)}{\partial t_{r,j}} &=  \frac{\partial }{\partial t_r} \sum_{l=1}^m \left|\sum_{i=1}^k a_i\alpha_l(t_i)-y_l\right|^2\\
 &=  \sum_{l=1}^m 2 \re \left( a_r \partial_j\alpha_l(t_r) \left(\overline{\sum_{i=1}^k a_i\alpha_l(t_i)-y_l}\right)\right)\\
 &= - 2 a_r\re \ls A \delta_{t_r,j}^{\prime}, A \phi(\theta)-y\rs.\\
 \end{split}
 \end{equation}
\end{proof}

\begin{proof}[Proof of Proposition~\ref{prop:Hessian}]
For $H_{1,r,s}$,
 \begin{equation}
  \begin{split}
   \frac{\partial^2 g(\theta)}{\partial a_r \partial a_s} &=  \frac{\partial }{\partial a_s } \sum_{l=1}^m 2 \re \left(  \alpha_l(t_r) \left(\overline{ \sum_{i=1}^k a_i\alpha_l(t_i)-y_l}\right)\right)\\
     &=   \sum_{l=1}^m 2 \re  \left( \alpha_l(t_r) \overline{ \alpha_l(t_s)}\right).\\
 \end{split}
 \end{equation}
For $H_{2,r,j_1,s,j_2}$,
 \begin{equation}
  \begin{split}
   \frac{\partial^2 g(\theta)}{\partial t_{r,j_1} \partial t_{s,j_2}} &=  \frac{\partial }{\partial t_{s,j_1}} \sum_{l=1}^m 2 \re   \left( a_r \partial_{j_1}\alpha_l(t_r) \left(\overline{ \sum_{i=1}^k a_i\alpha_l(t_i)-y_l}\right)\right)\\
&=   \sum_{l=1}^m 2 \re   \left(   a_r  \partial_{j_1}\alpha_l(t_r)  \left(\overline{ a_s \partial_{j_2}\alpha_l(t_s) }\right) \right) \\
&+ \indic(r=s) \sum_{l=1}^m 2 \re  \left( a_r \partial_{j_2}\partial_{j_1}\alpha_l(t_r) \left(\overline{ \sum_{i=1}^k a_i\alpha_l(t_i)-y_l}\right)\right).\\
 \end{split}
 \end{equation}
For $H_{12,r,s,j}$
 \begin{equation}
  \begin{split}
   \frac{\partial^2 g(\theta)}{\partial a_r \partial t_{s,j}} &=  \frac{\partial}{\partial t_{s,j}} \sum_{l=1}^m   2 \re (\alpha_l(t_r)) \left(\overline{ \sum_{i=1}^k a_i\alpha_l(t_i)-y_l}\right)\\
     &=   \sum_{l=1}^m 2 \re  \left( \alpha_l(t_r) \left(\overline{ a_s\partial_{j}\alpha_l(t_s)}\right)\right)\\
     &+ \indic(r=s) \sum_{l=1}^m 2 \re  \left(  \partial_{j}\alpha_l(t_r) \left(\overline{ \sum_{i=1}^k a_i\alpha_l(t_i)-y_l}\right)\right).\\
 \end{split}
 \end{equation}

\end{proof}
 \subsection{Proofs for Section~\ref{sec:kernel_dipole}} \label{proof22}
 
 \snewtext
 \begin{proof}[Proof of Lemma~\ref{def:scalar_dip}]
 First remark that  a generalized dipole $\nu=  a \delta_{t} + b\delta_{t,v}^{\prime}$ with $\|v\|_2=1$ is the limit in the distributional sense of the dipoles $\nu^{\eta} = a\delta_t - b\frac{\delta_{t+\eta v}-\delta_t}{\eta}$ when $\eta \to 0$.

Now let two generalized dipoles $\nu_1 = a_1 \delta_{t_1} +b_1 \delta_{t_1,v_1}^{\prime}, \nu_2 =a_2 \delta_{t_2} +b_2 \delta_{t_2,v_2}^{\prime}$. The $\nu_i$ are the limit (in the distributional sense) of a family of dipole $\nu_i^{\eta_i}$ for $\eta_i  \to 0^+$.  Let $f(t) = \rho(\|t\|_2)$. We have
\begin{equation}
\ls\nu_1^{\eta_1}, \nu_2^{\eta_2}\rs_h = \int_{\bR^d}  \int_{\bR^d}  f(t-s)\id\nu_1^{\eta_1}(t) \id\nu_2^{\eta_2}(s).
\end{equation}

Remark that by construction $ g_{\eta_1}(s) := \int_{\bR^d}   f(t-s) \id\nu_1^{\eta_1}(t) \to_{\eta_1 \to 0^+}  g(s) := a_1 f(t_1-s) + b_1\lsd \delta_{t_1,v_1}^{\prime}, f(\cdot-s) \rsd < +\infty$ where $g_{\eta_1}$ is in $\sC^2$ and $g$ is in $\sC^1$ thanks to the assumption on $h$ and $\rho$.
Hence by boundedness of the integrals and the dominated convergence theorem, for any $\eta_2$,
\begin{equation}
\ls\nu_1^{\eta_1}, \nu_2^{\eta_2}\rs_h  \to_{\eta_1 \to 0^+}  \int_{\bR^d} g(s) \id\nu_2^{\eta_2}(s).
\end{equation}
Moreover, by construction of $\nu_2^{\eta_2} $, and symmetry of $f$ (i.e. $f(t_1-t_2) = f(t_2-t_1)$),
\begin{equation}
\begin{split}
 \int_{\bR^d}  g(s) \id\nu_2^{\eta_2}(s)
 &\to_{\eta_2 \to 0^+} a_1 a_2f( t_1-t_2 ) \\
 & + a_2 b_1  \lsd \delta_{t_1,v_1}^{\prime},f(\cdot-t_2) \rsd +a_1b_2\lsd \delta_{t_2,v_2}^{\prime},f(t_1-\cdot) \rsd \\
 &-b_1b_2 \lsd  \delta_{t_2,v_2}^{\prime},  f_{v_1}^{\prime}(t_1- \cdot  ) \rsd\\
 & =  a_1 a_2f( t_1-t_2 )  - a_2 b_1  f_{v_1}^{\prime}(t_1-t_2) -a_1b_2f_{v_2}^{\prime}(t_2-t_1)\\
 &-b_1b_2  f_{v_1,v_2}^{\prime \prime}(t_1- t_2  )\\  
 \end{split}
\end{equation}
We define $\ls \nu_1, \nu_2\rs_h :=  a_1 a_2f( t_1-t_2 )  - a_2 b_1  f_{v_1}^{\prime}(t_1-t_2) -a_1b_2f_{v_2}^{\prime}(t_2-t_1) -b_1b_2  f_{v_1,v_2}^{\prime \prime}(t_1- t_2  ) $.
We just showed that  
\begin{equation}
\ls\nu_1^{\eta_1}, \nu_2^{\eta_2}\rs_h  \to_{\eta_1 \to 0^+,\eta_2 \to 0+} \ls \nu_1, \nu_2\rs_h .
\end{equation}

Note that the value of $\ls \nu_1, \nu_2\rs_h $ only depends on $\rho, \nu_1, \nu_2$.
 \end{proof}

\begin{proof}[Proof of Lemma \ref{lem:kernel_dirac_properties}]
 Using  Lemma~\ref{def:scalar_dip} with $t_1=t_2=t$, $b_1 = b_2 = 0$ and $a_1 =a_2 =1$  gives
  \begin{equation}
  \|\delta_t\|_h^2 =  \rho(0).
 \end{equation}
 Using  Lemma~\ref{def:scalar_dip} with $t_1=t_2=t$, $b_1  =a_2 = 0$ and $a_1 =b_2 =1$  gives
 
 \begin{equation}
\begin{split}
  \ls \delta_t, \delta_{t,v}^{\prime}\rs_h  := -f_{v}^{\prime}(0) = -\lim_{\eta\to 0^+ } \frac{\rho(\eta\|v\|)-\rho(0)}{\eta} = -\rho'(0)=0.
 \end{split}
 \end{equation}
  Using  Lemma~\ref{def:scalar_dip} with $t_1=t_2=t$, $b_1 = b_2 = 1$ and $a_1 =a_2 =0$  gives
 \begin{equation}
  \|\delta_{t,v}^{\prime}\|_h^2 := -f_v''(0)= |\rho''(0)|.
 \end{equation}
 
\end{proof}
\enewtext

\begin{proof}[Proof of Lemma \ref{lem:dip2Generalizeddip}]
\snewtext
Using the construction from the proof of Lemma~\ref{def:scalar_dip}, let two $\epsilon$-separated generalized dipole $\nu_1 , \nu_2$. The $\nu_i$ are the limit (in the distributional sense) of a family of $\epsilon$-separated dipole $\nu_i^{\eta_i}$ for $\eta_i  \to 0^+$. With the hypothesis, we have 
\begin{equation}\label{eq:dip2Generalizeddip1}
 \ls\nu_1^{\eta_1}, \nu_2^{\eta_2}\rs_h  \leq  \mu \|\nu_1^{\eta_1}\|_h  \|\nu_2^{\eta_2}\|_h.
\end{equation}
Furthermore, 
\begin{equation}
\ls\nu_1^{\eta_1}, \nu_2^{\eta_2}\rs_h  \to_{\eta_1 \to 0^+,\eta_2 \to 0+} \ls \nu_1, \nu_2\rs_h .
\end{equation}
\enewtext
Let $\nu=  a \delta_{t} + b\delta_{t,v}^{\prime}$ with $\|v\|_2=1$ and $\nu^{\eta} = a\delta_t - b\frac{\delta_{t+\eta v}-\delta_t}{\eta} =\left(a+\frac{b}{\eta}\right)\delta_t - b\frac{\delta_{t+\eta v}}{\eta}$.We have  $\|\nu\|_h^2 = a^2 +b^2|\rho''(0)| $ (with Lemma~\ref{lem:kernel_dirac_properties}) and 
\begin{equation}
\begin{split}
 \|\nu^\eta\|_h^2 &=\left(a+\frac{b}{\eta}\right)^2  +\left(\frac{b}{\eta}\right)^2 -2\left(a+\frac{b}{\eta}\right)\frac{b}{\eta}\rho(\eta) 
 \\ &
 = a^2 +2\left(\frac{b}{\eta}\right)^2 +2\frac{ab}{\eta} -2\frac{ab}{\eta}\rho(\eta) - 2\left(\frac{b}{\eta}\right)^2\rho(\eta)\\
 &= a^2 + 2\frac{ab}{\eta}(1-\rho(\eta))  +2\frac{b^2}{\eta^2}(1-\rho(\eta)).\\
 \end{split}
\end{equation}
 
But $\frac{1-\rho(\eta)}{\eta}=\frac{\rho(0)-\rho(\eta)}{\eta} \to -\rho'(0)$ when $\eta \to 0^+$, and $\rho'(0)=0$.

Moreover, $\rho(\eta)=h(0)+\eta \rho'(0) + \frac{\eta^2}{2}\rho''(0)+o(\eta^2)=1- \frac{\eta^2}{2}|\rho''(0)|+o(\eta^2)$. Hence $\frac{1-\rho(\eta)}{\eta^2}\to_{\eta \to 0^+}  \frac{1}{2}|\rho''(0)|$.

We thus deduce that $\|\nu^\eta\|_h^2  \to a^2 + b^2 |\rho''(0)|=\|\nu\|_h$ when $\eta \to 0^+$.

Hence, with such choice of $\nu_1^{\eta_1} $ $ \nu_2^{\eta_2}$, we can take the limit $\eta_1,\eta_2 \to 0$ in Equation~\eqref{eq:dip2Generalizeddip1} to get the result.

\end{proof}
\begin{proof}[Proof of Lemma~\ref{lem:mutual_Generalized_dipoles}]
Using Lemma~\ref{lem:dip2Generalizeddip}, and the same proof as in Lemma~\ref{lem:pyth_dipole}, we get the result.
 
\end{proof}

\begin{proof}[Proof of Lemma~\ref{lem:RIP_derivative}]
Let $\nu_r = a_r\delta_{t_r}+b_r\delta_{t_r,v}^{\prime}$ the $\epsilon$-separated generalized dipoles. Similarly to  Lemma~\ref{lem:dip2Generalizeddip}, take $\nu_r^\eta= (a_r+\frac{b_r}{\eta})\delta_{t_r} - b_r\frac{\delta_{t_r +\eta v}}{\eta}$. For sufficiently small $\eta$ the $ \nu_r^\eta$ are $\epsilon$-separated dipoles, hence $\sum \nu_r^\eta \in \Sigma-\Sigma$ and 

 \begin{equation} \label{eq:RIP_int_dirac_deriv}
\begin{split}
 (1-\RIPcst) \left\|  \sum_{r=1,k}\nu_r^\eta \right\|_h^2 &\leq \left\|A(  \sum_{r=1,k}\nu_r^\eta)\right\|_2^2   
 \leq (1+\RIPcst)\left\|  \sum_{r=1,k}\nu_r^\eta \right\|_h^2.  \\
\end{split}
 \end{equation}

 Now remark that $g_1(\eta)=\|  \sum_{r=1,k}\nu_r^\eta \|_h^2$ and $g_2(\eta)=\|A(  \sum_{r=1,k}\nu_r^\eta )\|_2^2 $  are continuous functions of $\eta$ that converge to $\| \sum_{r=1,k} (a_r\delta_{t_r}+b_r\delta_{t_r,v}^{\prime})\|_h^2 $ and  $\|A(\sum_{r=1,k} (a_r\delta_{t_r}+b_r\delta_{t_r,v}^{\prime}))\|_2^2  $ when $\eta \to 0$:
 \begin{itemize}
 \item For $g_1$, use the same proof as in Lemma~\ref{lem:dip2Generalizeddip}
 with the linearity of the limit. 
  \item For $g_2$:
  \begin{equation}
  \begin{split}
g_2(\eta)  &= \sum_{l=1,m} \left|\sum_{r=1,k} \int \alpha_l(t)(a_r \id\delta_{t_r}(t)-\frac{b_r}{\eta}(\id\delta_{t_r+\eta v}(t)-\id\delta_{t_r}(t)))\right|^2\\
 &= \sum_{l=1,m} \left|\sum_{r=1,k} \left( \alpha_l(t_r)a_r 
 -\frac{b_r}{\eta}(\alpha_l(t_r+\eta v)-\alpha_l(t_r))  \right)\right|^2\\
  &\to_{\eta\to 0^+} \sum_{l=1,m} \left|\sum_{r=1,k} \left( \alpha_l(t_r)a_r 
  -b_r(\alpha_l)_v'(t_r)\right)\right|^2\\
&=\left\|A(\sum_{r=1,k} a_r\delta_{t_r}+b_r\delta_{t_r,v}^{\prime})\right\|_2^2. 
\end{split}
  \end{equation}

  Taking  the limit of Equation~\eqref{eq:RIP_int_dirac_deriv} for $\eta \to 0$ yields the result.
 \end{itemize}
 
\end{proof}

\begin{proof}[Proof of Lemma~\ref{lem:upper_RIP_second_deriv}]

We have
 \begin{equation}
 \begin{split}
   \|A\delta_{t,v_1,v_2}^{\prime \prime}\|_2^2 & = \sum_{l=1,m} | (A \delta_{t,v_1,v_2}^{\prime \prime})_l|^2\\
   & = \sum_{l=1,m} | \alpha_{l,v_1,v_2}^{\prime \prime}(t)|^2\\
   &\leq m \sup_{l=1,m; t \in \sB_2(R)} | \alpha_{l,v_1,v_2}^{\prime \prime}(t)|^2 \leq m D_{A,R}^2
   \end{split}
 \end{equation}
 
 where $D_{A,R}$ is given in \eqref{defDAR}, i.e. $D_{A,R}$ is the supremum of directional second  derivatives of the $\alpha_l$ over 
$\sB_2(R)$. We have $D_{A,R} <+\infty$ because the $\alpha_l$ are supposed to be in $\sC^2(\sB_2(R))$.

\end{proof}
\begin{lemma}\label{lem:rkhs_convol}
 Let $K$ be a symmetrical convolution kernel in $\sC^2$ and $h_K: (t_1,t_2) \to h_K(t_1,t_2) = [K*K](t_1-t_2)$ (the convolution of $K$ by itself) then  for any $x \in \Sigma_{k,\epsilon}-\Sigma_{k,\epsilon}$, we have 

 \begin{equation}
  \|x\|_{h_K}^2 = \|K* x\|_{L^2}^2.
 \end{equation}

\end{lemma}
\begin{proof}[Proof of Lemma \ref{lem:rkhs_convol}]
 Write $x = \sum a_i \delta_{t_i}$ and use the symmetry of $K$:
\begin{equation}
\begin{split}
 \|K*x\|_{L^2}^2  = \int \left|\sum a_i K(t-t_i) \right|^2 \id t  &= \sum_{i,j} a_i a_j \int K(t-t_i)K(t-t_j) \id t\\
 &=  \sum_{i,j} a_i a_j \int K(t)K(t+t_i-t_j) \id t\\
 &=  \sum_{i,j} a_i a_j [K*K] (t_i-t_j) =  \|x\|_{h_K}^2.
 \end{split}
\end{equation}

\end{proof}

 \subsection{Proofs for Section~\ref{sec:control_Hessian}} \label{proof23}

We will use the following Lemma on directional derivatives of Diracs. 

\begin{lemma}\label{lem:sum_directional_dirac_derivative}
 Let $u, t_0\in \bR^d$. Suppose $u \neq 0$. Then, $\sum_{i=1,d} u_i \delta_{t_{0},j}^{\prime} = \|u\|_2\delta_{t_{0},\frac{u}{\|u\|_2}}^{\prime}$.
\end{lemma}
\begin{proof}
Let $f$ a function in $\sC^2(\bR^d)$, we have

$\int_{t\in\bR^d}f(t) \sum_{i=1,d} u_i \id\delta_{t_{0},i}^{\prime}(t) = -\sum_{i=1,d} u_i\partial_i f(t_0) = -\ls u_i, \nabla f(t_0) \rs = -\|u\|_2f_{\frac{u}{\|u\|_2}}'(t_0) $. Hence, $\sum_{i=1,d} u_i \delta_{t_{0},i}^{\prime} = \|u\|_2\delta_{t_{r},\frac{u}{\|u\|_2}}^{\prime}$
\end{proof}

To prove Theorem~\ref{th:min_max_eigen_control_H}, we control first the eigenvalues of $G$ in the decomposition $H = G +F$.

\begin{lemma}\label{le:min_max_eigen_control_G}
 Suppose $h$ follows Assumption~\ref{assum:kernel_prop}. Let $\theta = (a_1,..,a_k, t_1,..t_k) \in \Theta_{k,\epsilon}$ with $t \in \text{rint}\sB_2(R)$. Let $H$ the Hessian of $g$ at $\theta$. Suppose $A$ has RIP $\gamma$ on $\Sigma_{k,\epsilon}-\Sigma_{k,\epsilon}$. We have 
 \begin{equation}
 \sup_{\|u\|_2 =1} u^TGu \leq  2(1+\RIPcst)(1+(k-1)\mu)\max(1,(a_r^2|\rho''(0)|)_{r=1,l}); \\
\end{equation}
\begin{equation}
 \inf_{\|u\|_2 =1} u^TGu \geq 2(1-\RIPcst)(1-(k-1)\mu)\min(1,(a_r^2|\rho''(0)|)_{r=1,l}).\\
\end{equation}
where $G$ is defined in Proposition~\ref{prop:Hessian}.
\end{lemma}

\begin{proof}
Let $u \in \bR^{k(d+1)} $ such that $\|u\|_2=1$. We index $u$ as follows: $u_r \in \bR$ for $r=1,k$. $u_r\in\bR^{d}$ for $r=k+1,2k$ (it follows the indexing of $H$ and $G$ we used). Remark that 
\begin{equation}
 \begin{split}
  u^TGu =& \sum_{r,s=1,k} u_r u_s G_{1,r,s} + \sum_{r=k+1,2k;j_1=1,d;s=k+1,2k;j_2=1,d} u_{r,j_1} u_{s,j_2} G_{2,r,j_1,s,j_2}\\
  &+  \sum_{r=1,k;s=k+1,2k;j=1,d} u_r u_{s,j} G_{12,r,s,j} +  \sum_{r=k+1,2k;j=1,d;s=1,k} u_{r,j} u_s G_{21,r,j,s}\\
  =&  2\sum_{r,s=1,k} \re\ls Au_r\delta_{t_r}, Au_s\delta_{t_s} \rs \\
   & + 2\sum_{r=k+1,2k;j_1=1,d;s=k+1,2k;j_2=1,d} \re\ls Au_{r,j_1}a_{r-k}\delta_{t_{r-k},j_1}^{\prime}, Au_{s,j_2}a_{s-k}\delta_{t_{s-k},j_2}^{\prime} \rs \\
  &-2\sum_{r=1,k;s=k+1,2k;j=1,d}  \re\ls Au_r\delta_{t_r}, Au_{s,j}a_{s-k}\delta_{t_{s-k},j}^{\prime} \rs
  \\ & - 2\sum_{r=k+1,2k;j=1,d;s=1,k} \re\ls Au_{r,j}a_{r-k}\delta_{t_{r-k},j}^{\prime}, Au_s\delta_{t_{s}} \rs 
   \end{split} 
\end{equation}
Thus we have
\begin{equation}
 \begin{split}
 u^TGu 
  =&  2\left\|A\sum_{r=1,k} u_r \delta_{t_r}\right\|_2^2 +  2\left\|A\sum_{r=k+1,2k;j=1,d} u_{r,j} a_{r-k}\delta_{t_{r-k},j}^{\prime}\right\|_2^2  \\
  &-2\re  \left\ls A\sum_{r=1,k} u_r \delta_{t_r}, A\sum_{r=k+1,2k;j=1,d} u_{r,j} a_{r-k}\delta_{t_{r-k},j}^{\prime}\right\rs  \\
  &-2\re  \left\ls A\sum_{r=k+1,2k;j=1,d} u_{r,j} a_{r-k}\delta_{t_{r-k},j}^{\prime}, A\sum_{r=1,k} u_r \delta_{t_{r}}\right\rs \\
   =&  2\left\|A\left(  \sum_{r=1,k}u_r\delta_{t_r}- \sum_{r=k+1,2k;j=1,d} u_{r,j}a_{r-k}\delta_{t_{r-k},j}^{\prime}\right) \right\|_2^2 \\
   =&  2\left\|A \left( \sum_{r=1,k}\left( u_r\delta_{t_r}- a_{r}\sum_{j=1,d}u_{r+k,j}\delta_{t_{r},j}^{\prime}\right)\right) \right\|_2^2. \\
 \end{split} 
\end{equation}

Using Lemma~\ref{lem:sum_directional_dirac_derivative}, we have $\sum_{j=1,d} w_j \delta_{t_{r},j}^{\prime} = \|w\|_2\delta_{t_{r},\frac{w}{\|w\|_2}}^{\prime}$ and 

\begin{equation}\label{eq:expr_G}
 \begin{split}
  u^TGu =& 2\left\|A\sum_{r=1,k} (u_r\delta_{t_r}- a_{r}\|u_{r+k}\|_2\delta_{t_{r},\frac{u_{r+k}}{\|u_{r+k}\|_2}}^{\prime})\right\|_2^2. \\
 \end{split} 
\end{equation}

We use the lower RIP in Lemma~\ref{lem:RIP_derivative},
\begin{equation}
 \begin{split}
  u^TGu & \geq  2(1-\RIPcst)\left\|\sum_{r=1,k}( u_r\delta_{t_r} - a_{r}\|u_{r+k}\|_2\delta_{t_{r},\frac{u_{r+k}}{\|u_{r+k}\|_2}}^{\prime})\right\|_h^2 .\\
   \end{split} 
\end{equation}

   Then the hypothesis on $\|\cdot\|_h$ and  Lemma~\ref{lem:mutual_Generalized_dipoles}  yields
   
     \begin{equation}
   \begin{split}
   \|\sum_{r=1,k}( u_r\delta_{t_r} - a_{r}\|u_{r+k}\|_2\delta_{t_{r},\frac{u_{r+k}}{\|u_{r+k}\|_2}}^{\prime})\|_h^2 &
   \\
   \geq (1-(k-1)\mu) \sum_{r=1,k} &\| u_r\delta_{t_r} - a_{r}\|u_{r+k}\|_2\delta_{t_{r},\frac{u_{r+k}}{\|u_{r+k}\|_2}}^{\prime}\|_h^2
   \end{split}
\end{equation}
and 
  \begin{equation}
 \begin{split}
  u^TGu &\geq  2(1-\RIPcst)(1-(k-1)\mu) \sum_{r=1,k} \| u_r\delta_{t_r}- a_{r}\|u_{r+k}\|_2\delta_{t_{r},\frac{u_{r+k}}{\|u_{r+k}\|_2}}^{\prime}\|_h^2\\
  & \geq  2(1-\RIPcst)(1-(k-1)\mu) \sum_{r=1,k} \left(| u_r|^2 -2a_r u_r \|u_{k+r}\|_2\ls\delta_{t_r},\delta_{t_{r},\frac{u_{r+k}}{\|u_{r+k}\|_2}}^{\prime} \rs_h
  \right.
  \\ &  \left. + a_r^2\|u_{k+r}\|_2^2 \|\delta_{t_{r},\frac{u_{r+k}}{\|u_{r+k}\|_2}}^{\prime}\|_h^2\right).\\
  \end{split} 
\end{equation}

Then using Lemma~\ref{lem:kernel_dirac_properties}:
    \begin{equation}
 \begin{split}
 u^TGu  & \geq  2(1-\RIPcst)(1-(k-1)\mu) \sum_{r=1,k} \left(| u_r|^2  + a_r^2 \|u_{k+r}\|_2^2 |\rho''(0)| \right)\\
  &\geq  2(1-\RIPcst)(1-(k-1)\mu) \inf_{\|u\|_2=1} \sum_{r=1,k} \left(| u_r|^2  + \| u_{k+r}\|_2^2a_r^2|\rho''(0)| \right). \\
  & =  2(1-\RIPcst)(1-(k-1)\mu)\min(1,(a_r^2|\rho''(0)|)_{r=1,l}) .\\
 \end{split} 
\end{equation}

Similarly, using the upper RIP in Lemma~\ref{lem:RIP_derivative}:

\begin{equation}
 \begin{split}
 u^TGu & \leq  2(1+\RIPcst)\|\sum_{r=1,k}( u_r\delta_{t_r}-  a_{r}\|u_{r+k}\|_2\delta_{t_{r},\frac{u_{r+k}}{\|u_{r+k}\|_2}}^{\prime})\|_h^2 .\\
   \end{split} 
\end{equation}

   Then the hypothesis on $\|\cdot\|_h$  yields (Lemma~\ref{lem:mutual_Generalized_dipoles}) 

   \begin{equation}
   \begin{split}
   \|\sum_{r=1,k}( u_r\delta_{t_r}-  a_{r}\|u_{r+k}\|_2\delta_{t_{r},\frac{u_{r+k}}{\|u_{r+k}\|_2}}^{\prime})\|_2^2 &
   \\ \leq (1+(k-1)\mu) \sum_{r=1,k}& \| u_r\delta_{t_r}\modif{-  a_{r}\|u_{r+k}\|_2\delta_{t_{r},\frac{u_{r+k}}{\|u_{r+k}\|_2}}^{\prime}}\|_h^2 \\
   \end{split}
\end{equation}

   and 
  \begin{equation}
 \begin{split}
  u^TGu &\leq  2(1+\RIPcst)(1+(k-1)\mu) \sum_{r=1,k} \| u_r\delta_{t_r}\modif{- a_{r}\|u_{r+k}\|_2\delta_{t_{r},\frac{u_{r+k}}{\|u_{r+k}\|_2}}^{\prime}}\|_h^2.\\
  \end{split} 
\end{equation}
Then using Lemma~\ref{lem:kernel_dirac_properties}:
    \begin{equation}
 \begin{split}
 u^TGu  & \leq 2(1+\RIPcst)(1+(k-1)\mu) \sum_{r=1,k} \left(| u_r|^2  + a_r^2 \|u_{k+r}\|_2^2
 |\rho''(0)|  \right) \\
  &\leq  2(1+\RIPcst)(1+(k-1)\mu) \sup_{\|u\|_2=1} \sum_{r=1,k} \left(| u_r|^2  + \| u_{k+r}\|_2^2a_r^2|\rho''(0)| \right) \\
  &=  2(1+\RIPcst)(1+(k-1)\mu)\max(1,(a_r^2|\rho''(0)|)_{r=1,l}). \\
 \end{split} 
\end{equation}
 \end{proof}

\begin{proof}[Proof of Theorem \ref{th:min_max_eigen_control_H}]
Let $\theta^*$ a minimizer of \eqref{eq:minimization2}. Consider $H$ the Hessian of $g$ at $\theta$.  
We recall that $H=G+F$ (see Proposition~\ref{prop:Hessian}).
Using Lemma~\ref{le:min_max_eigen_control_G}, we just need to bound the operator norm of $F$ and then to combine it with the bounds on the eigenvalues of $G$  to get bounds on eigenvalues of $H=G+F$.

We use  Lemma~\ref{lem:upper_RIP_second_deriv}, the Cauchy-Schwartz  and triangle inequalities. We have $\|A \delta_{t_r,j_1,j_2}^{\prime \prime} \|_2 \leq \sqrt{m}D_{A,R} $   and 

\begin{equation}
 \begin{split}
 |F_{2,r,j_1,s,j_2}| & \leq\indic(r=s)2|a_r|\| A \delta_{t_r,j_1,j_2}^{\prime \prime}\|_2  \|A\phi(\theta)-y\|_2 .\\
  & \leq\indic(r=s)2|a_r|\sqrt{m}D_{A,R}\|A\phi(\theta)-A\phi(\theta^*) + A\phi(\theta^*)-y\|_2  .\\
 & \leq\indic(r=s)2|a_r|\sqrt{m}D_{A,R} (\|A\phi(\theta)-A\phi(\theta^*)\|_2+ \|e\|_2).\\
 \end{split}
\end{equation}
Similarly, with Lemma~\ref{lem:RIP_derivative}, 
\begin{equation}
 \begin{split}
 F_{12,r,s,j} &\leq \indic(r=s) 2\sqrt{1+\RIPcst}\|\delta_{t_r,j}^{\prime}\|_h  \|A\phi(\theta)-y\|_2  \\
  &\leq \indic(r=s) 2\sqrt{1+\RIPcst}\sqrt{|\rho''(0)|} (\|A\phi(\theta)-A\phi(\theta^*)\|_2+ \|e\|_2). \\
 \end{split}
\end{equation}
Let $\|\cdot\|_{op}$ be the $\ell^2$ operator norm of a matrix. With Gerschgorin circle theorem~\cite{Golub_2012}, we have 
\begin{equation}
 \|F\|_{op}  \leq  \max_{l}  \|F_{l,:}\|_1
\end{equation}
where $F_{l,:}$ is the $l$-th row of $F$. We get
\begin{equation}
\begin{split}
 \|F\|_{op}  &\leq  \max( d\max_{r,s,j}   |F_{12,r,s,j}| ,  \max_{r,s,j}   |F_{12,r,s,j}|  + d \max_{r,j_1,s,j_2}   |F_{2,r,j_1,s,j_2} | )\\
             &\leq  (d+1)\max( \max_{r,s,j}   |F_{12,r,s,j}|  ,  \max_{r,j_1,s,j_2}   |F_{2,r,j_1,s,j_2} |)\\
             &\leq  2(d+1)\max( \max_{r}|a_r|  \sqrt{m}D_{A,R} ,\sqrt{1+\RIPcst}\sqrt{|\rho''(0)|} ) (\|A\phi(\theta)-A\phi(\theta^*)\|_2+ \|e\|_2).\\
 \end{split}
\end{equation}

Hence, using Weyl's perturbation inequalities on $H = G + F$, i.e. $\lambda_{min}(H) \geq \lambda_{min}(G)-\lambda_{max}(F)$ and  $\lambda_{max}(H) \leq \lambda_{max}(G) + \lambda_{max}(F)$, we get the result.

\end{proof}

\begin{proof}[Proof of Corollary~\ref{cor:control_Hessian}]

First, observe that at $\theta_0$, $F =0$.

The upper bound is a direct consequence of Theorem~\ref{le:min_max_eigen_control_G}. 

 We show the result in the case $ \max(1,(a_r^2|\rho''(0)|)_{r=1,l}) \neq 1$ and $ \min(1,(a_r^2|\rho''(0)|)_{r=1,l}) \neq 1 $ (the proof is similar in the other case). For the lower bound let $v \in \bR^{k(d+1)}$ and $i_0 = \arg \max_{r=1,l} (a_r^2|\rho''(0)|)$, set $\|v_{i_0}\|_2 =1$    and $v_j = 0$ for $j \neq i_0$. With Equation~\eqref{eq:expr_G}, we have 

\begin{equation}
 \sup_{\|u\|_2 =1} u^THu \geq  v^THv \geq 2(1-\RIPcst) \max(1,(a_r^2|\rho''(0)|)_{r=1,l}).
\\
\end{equation}
Similarly, let $v \in \bR^{k(d+1)}$ and $i_0 = \arg \min((a_r^2|\rho''(0)|)_{r=1,l})$, $\|v_{i_0}\|_2 =1$    and $v_j = 0$ for $j \neq i_0$.

\begin{equation}
 \inf_{\|u\|_2 =1} u^THu \leq 2(1+\RIPcst)\min(1,(a_r^2|\rho''(0)|)_{r=1,l}).\\
\end{equation}
\end{proof}

\subsection{Proofs for Section~\ref{sec:basin}} \label{proof3}
  
\begin{proof}[Proof of Theorem \ref{th:basin}]
 
 Let $\theta^* =(a_1,...,a_k,t_1,..t_k)\in \Theta_{k,\epsilon}$ the global minimum of $g$ and $\theta = (b_1,...,b_k,s_1,..s_k) \in \Lambda_{\theta^*,\beta} $. 
 
 \snewtext
 First notice that $ \|\theta-\theta^*\|_2^2 \leq \beta^2$ implies that for any $j$, we have $|a_j- b_j|^2 \leq \beta^2 $  and
 
 \begin{equation} \label{ineq:th2}
 |a_1|-\beta  \leq |a_j|-\beta \leq |b_j| \leq |a_j|+\beta\leq |a_k|+\beta.
 \end{equation}

 We also have  $\|s_j- t_j\|_2 <\beta  \leq \frac{\epsilon}{4}$.
 \enewtext
  Hence for $i\neq j$ we have $\|s_i-s_j\|_2 = \|s_i -t_i +t_i -t_j +t_j-s_j\|_2 \geq \|t_i -t_j\|_2 -\|t_i-s_i\|_2 -\|t_j-s_j\|_2 > \epsilon - 2\epsilon/4 = \epsilon/2$ and $\phi(\theta) \in \Sigma_{k,\frac{\epsilon}{2}}$. 
 
 We use Theorem~\ref{th:min_max_eigen_control_H} to get the bound on the min and max eigenvalues of the Hessian. 
 
 We can then plug  Inequality~\eqref{ineq:th2} into the one of Theorem~\ref{th:min_max_eigen_control_H}.

 Finally we notice the fact that $\sup_{\theta \in \Lambda_{\theta^*,\beta}} \|A\phi(\theta)-A\phi(\theta^*)\|_2$ exists because $\Lambda_{\theta^*,\beta}$ is bounded.

 \end{proof}

 \begin{proof}[Proof of Theorem \ref{th:basin_with_constraint}]
  This is a direct consequence of Theorem~\ref{th:min_max_eigen_control_H}. The  proof follows the same lines as the one of Theorem~\ref{th:basin}.
 \end{proof}

 \begin{proof}[Proof of Corollary \ref{cor:basin_noiseless}]
The set  $\Lambda =\Lambda_{\theta^*,\beta}$ is an open set where the Hessian of $g$ at $ \Lambda$ is positive as long as $\xi \leq 2 (1-\RIPcst)(1-(k-1)\mu)\min(1,(|a_1|-\beta)^2|\rho''(0)|)$ with  Theorem~\ref{th:basin}. 

In this case $g$ is  convex on $\Lambda$. Theorem~\ref{th:basin} also gives a uniform bound for the operator norm of the Hessian:  $\|H\|_{op} \leq 2(1+\RIPcst)(1+(k-1)\mu)\max(1,(|a_k|+\beta)^2|\rho''(0)|)+\xi$
and $g$ has Lipschitz gradient. \modif{Moreover the gradient descent on Lipschitz smooth convex functions guarantees that  $\|\theta_n -\theta^*\|_2$ decreases, where $\theta_n$ are the iterates of the gradient descent  (this is proved using direct consequences of the nonexpensiveness of  $\ls \tau \nabla g(\theta),  \cdot \rs$ \cite[Proposition 4.2 (iv)]{Bauschke_2011}). Hence the iterates $\theta_n$ stay in $\Lambda $} and we  deduce from  Corollary~\ref{cor:convergence_gradient_descent} that 
$\Lambda$
is a basin of attraction.

Hence we just need to show that  $\xi \leq 2(1-\RIPcst)(1-(k-1)\mu)\min(1,(|a_1|-\beta)^2|\rho''(0)|)$.
 Let $\theta \in \Lambda$, we have, with the RIP hypothesis,
 
 \snewtext
 
\begin{equation}
 \begin{split}
  \xi(\theta)&:=2(d+1)\max( \max_{r}|a_r|  \sqrt{m}D_{A,R} ,\sqrt{1+\RIPcst}\sqrt{|\rho''(0)|} ) \|A\phi(\theta)-A\phi(\theta^*)\|_2 \\
   \leq& 2(d+1)\max(  |a_k|  \sqrt{m}D_{A,R} ,\sqrt{1+\RIPcst}\sqrt{|\rho''(0)|} ) \sqrt{1+\RIPcst} \|\phi(\theta)-\phi(\theta^*)\|_h  \\
     \leq&  2(d+1)\max( |a_k|  \sqrt{m}D_{A,R} ,\sqrt{1+\RIPcst}\sqrt{|\rho''(0)|} ) \sqrt{1+\RIPcst} \sqrt{1+(k-1)\mu} \sqrt{\sum_i\|a_i\delta_{t_i}-b_i\delta_{s_i}\|_h^2}  \\
    \end{split}
\end{equation} 
where we wrote $\theta^* = \sum_i a_i \delta_{t_i}$ and $\theta = \sum_i b_i \delta_{s_i}$ such that $|s_i-t_i| \leq \epsilon/4$. 

We now bound the term $\sum_i\|a_i\delta_{t_i}-b_i\delta_{s_i}\|_h^2$:    

\begin{equation}
 \begin{split}
     \sum_i\|a_i\delta_{t_i}-b_i\delta_{s_i}\|_h^2 = & \sum_i a_i^2 +b_i^2 -2 a_i b_i\rho(\|s_i-t_i\|_2) \\
         =&   \sum_i \rho(\|s_i-t_i\|_2) |a_i- b_i|^2 + (1-\rho(\|s_i-t_i\|_2) )\sum_i a_i^2 +b_i^2 \\
         \end{split}
\end{equation}
Using the hypothesis that $\|\theta-\theta^*\|^2 \leq \beta^2$ and  $\beta \leq |a_1|/2$, we have $|b_i| \leq |a_i| + \beta \leq \frac{3}{2} |a_i|$.  With  the assumption on $h$ (and $\rho$),
   \begin{equation}
 \begin{split}      
   \sum_i\|a_i\delta_{t_i}-b_i\delta_{s_i}\|_h^2 \leq &   \beta^2 +  \frac{|\rho''(0)|}{2}\beta^2 \frac{13}{4}\|a^*\|_2^2  \\
        \leq &  \beta^2 + \frac{|\rho''(0)|}{2}\beta^2 4\|a^*\|_2^2 \\
         \leq & \beta^2( 1 + 2|\rho''(0)| \|a^*\|_2^2 ) \\
  \end{split}
\end{equation}
where $a^* = (a_1,...,a_k)$.
The fact that $\beta \leq |a_1|/2$ implies 
\begin{equation}
 \begin{split}
  &\frac{\xi(\theta)}{\min(1,(|a_1|-\beta)^2|\rho''(0)|)}\\
  & \leq  \frac{2  (d+1)\sqrt{1+\RIPcst}  \sqrt{1+(k-1)\mu} \max( |a_k|  \sqrt{m}D_{A,R} ,\sqrt{1+\RIPcst}\sqrt{|\rho''(0)|} )  \sqrt{ 1 + 2|\rho''(0)|\|a^*\|_2^2 }\beta }{\min(1,|a_1|^2|\rho''(0)|/4)}
\end{split}
\end{equation}

 Hence using the hypothesis that $$\beta \leq \frac{ (1-\RIPcst)(1-(k-1)\mu)\min(1,|a_1|^2|\rho''(0)|/4) }{ (d+1)\sqrt{1+\RIPcst}  \sqrt{1+(k-1)\mu}   \max( |a_k|  \sqrt{m}D_{A,R} ,\sqrt{1+\RIPcst}\sqrt{|\rho''(0)|} )  \sqrt{1 + 2|\rho''(0)|\|a^*\|_2^2}} $$ we have
\enewtext

\begin{equation}
\begin{split}
 \xi(\theta)& \leq 2 (1-\RIPcst)(1-(k-1)\mu)\min(1,(|a_1|(1-\beta))^2|\rho''(0)|).
\end{split}
\end{equation}

 \end{proof}

 \begin{proof}[Proof of Corollary \ref{cor:basin_noisy}]
Following the same argument as Corollary \ref{cor:basin_noisy},   we just need to show that  $\xi \leq 2(1-\RIPcst)(1-(k-1)\mu)\min(1,(|a_1|-\beta))^2|\rho''(0)|)$. Let $\theta \in \Lambda_{\theta^*,\beta}$, we have, with the RIP hypothesis,

\begin{equation}
 \begin{split}
  \xi(\theta)&:=2(d+1)\max( \max_{r}|a_r|  \sqrt{m}D_{A,R} ,\sqrt{1+\RIPcst}\sqrt{|\rho''(0)|} )( \|A\phi(\theta)-A\phi(\theta^*)\|_2 +\|e\|_2)\\
   \leq& 2(d+1)\max(  |a_k|  \sqrt{m}D_{A,R} ,\sqrt{1+\RIPcst}\sqrt{|\rho''(0)|} ) (\sqrt{1+\RIPcst} \|\phi(\theta)-\phi(\theta^*)\|_h +\|e\|_2) \\
     \leq&  2(d+1)\max( |a_k|  \sqrt{m}D_{A,R} ,\sqrt{1+\RIPcst}\sqrt{|\rho''(0)|} ) (\sqrt{1+\RIPcst} \sqrt{1+(k-1)\mu} \sqrt{\sum_i\|a_i\delta_{t_i}-b_i\delta_{s_i}\|_h^2}+\|e\|_2)  \\
    \end{split}
\end{equation} 
where we wrote $\theta^* = \sum_i a_i \delta_{t_i}$ and $\theta = \sum_i b_i \delta_{s_i}$ such that $|s_i-t_i| \leq \epsilon/4$. 

Similarly as  in Corollary \ref{cor:basin_noisy}, we bound the term $\sum_i\|a_i\delta_{t_i}-b_i\delta_{s_i}\|_h^2$:    
\begin{equation}
 \begin{split}
     \sum_i\|a_i\delta_{t_i}-b_i\delta_{s_i}\|_h^2   \leq & \beta^2( 1 + 2|\rho''(0)| \|a^*\|_2^2 ) \\
  \end{split}
\end{equation}
The fact that $\beta \leq |a_1|/2$  and  $\|e\|_2  \leq  \sqrt{1+\RIPcst} \sqrt{1+(k-1)\mu} \beta$ implies 
\begin{equation}
 \begin{split}
  &\frac{\xi(\theta)}{\min(1,(|a_1|-\beta)^2|\rho''(0)|)}\\
  & \leq  \frac{2  (d+1)\sqrt{1+\RIPcst}  \sqrt{1+(k-1)\mu} \max( |a_k|  \sqrt{m}D_{A,R} ,\sqrt{1+\RIPcst}\sqrt{|\rho''(0)|} ) (1+ \sqrt{ 1 + 2|\rho''(0)| \|a^*\|_2^2 })\beta }{\min(1,|a_1|^2|\rho''(0)|/4)}
\end{split}
\end{equation}

 Hence using the hypothesis that $$\beta \leq \frac{ (1-\RIPcst)(1-(k-1)\mu)\min(1,|a_1|^2|\rho''(0)|/4) }{ (d+1)\sqrt{1+\RIPcst}  \sqrt{1+(k-1)\mu}   \max( |a_k|  \sqrt{m}D_{A,R} ,\sqrt{1+\RIPcst}\sqrt{|\rho''(0)|} )  (1+ \sqrt{ 1 + 2|\rho''(0)| \|a^*\|_2^2 })}, $$ we have
\enewtext

\begin{equation}
\begin{split}
 \xi(\theta)& \leq 2 (1-\RIPcst)(1-(k-1)\mu)\min(1,(|a_1|(1-\beta))^2|\rho''(0)|).
\end{split}
\end{equation}

 \end{proof}
 
\subsection{Proofs for Section~\ref{sec:projected_gradient}} \label{proof4}
 
\begin{proof}[Proof of Lemma \ref{lem:pseudo_conv}]
 Remark that $g(\theta)$ does not depend on the ordering of the positions. Reorder $\theta_0 =(a,t )$ and $\theta_1=(b,s)$ such that $t_1< t_2...<t_k$ and $s_1<s_2...<s_k$. Consider the function $g_1(\lambda) = g(\theta_\lambda)$ with $\theta_\lambda = (1-\lambda) \theta_0 + \lambda \theta_1$. Remark that $g_1$ is a continuous function of $\lambda$ taking values $g_1(0)= g(\theta_0)$ and $g_1(1)=g(\theta_1)$. Hence, with the intermediate value theorem, there is $\lambda$ such that $g(\theta_\lambda)=g_1(\lambda) = \alpha$. Moreover, denoting $\theta_\lambda = (a_\lambda, t_\lambda)$, we have, using the sorting of $t$ and $s$, for $1  \leq i < k $,
 \begin{equation}
 \begin{split}
  |t_{\lambda,i+1} -t_{\lambda,i}| &= |(1-\lambda)t_{i+1}  + \lambda s_{i+1}   -(1-\lambda)t _{i} - \lambda s_{i}| \\
  &= (1-\lambda)|t_{i+1}-t_i| +\lambda |s_{i+1}-s_i| > (1-\lambda)\epsilon + \lambda \epsilon = \epsilon.\\
   \end{split}
 \end{equation}
Hence $\theta_\lambda \in \Theta_{k,\epsilon}$.
\end{proof}
\bibliographystyle{abbrv}
 \bibliography{working_paper_SR}

\begin{thebibliography}{10}

\bibitem{Bauschke_2011}
H.~H. Bauschke and P.~L. Combettes.
\newblock {\em Convex analysis and monotone operator theory in Hilbert spaces},
  volume 408.
\newblock Springer, 2011.

\bibitem{Bhaskar_2013}
B.~N. Bhaskar, G.~Tang, and B.~Recht.
\newblock Atomic norm denoising with applications to line spectral estimation.
\newblock {\em IEEE Transactions on Signal Processing}, 61(23):5987--5999,
  2013.

\bibitem{Bhojanapalli_2016}
S.~Bhojanapalli, B.~Neyshabur, and N.~Srebro.
\newblock Global optimality of local search for low rank matrix recovery.
\newblock In {\em Advances in Neural Information Processing Systems}, pages
  3873--3881, 2016.

\bibitem{Blumensath_2011}
T.~Blumensath.
\newblock Sampling and reconstructing signals from a union of linear subspaces.
\newblock {\em IEEE Transactions on Information Theory}, 57(7):4660--4671,
  2011.

\bibitem{Bourrier_2014}
A.~Bourrier, M.~Davies, T.~Peleg, P.~Perez, and R.~Gribonval.
\newblock Fundamental performance limits for ideal decoders in high-dimensional
  linear inverse problems.
\newblock {\em Information Theory, IEEE Transactions on}, 60(12):7928--7946,
  2014.

\bibitem{Boyer_2017}
C.~Boyer, Y.~De~Castro, and J.~Salmon.
\newblock Adapting to unknown noise level in sparse deconvolution.
\newblock {\em Information and Inference: A Journal of the IMA}, 6(3):310--348,
  2017.

\bibitem{Cambareri_2018}
V.~Cambareri and L.~Jacques.
\newblock Through the haze: a non-convex approach to blind gain calibration for
  linear random sensing models.
\newblock {\em Information and Inference: A Journal of the IMA}, 2018.

\bibitem{Candes_2013}
E.~J. Cand{\`e}s and C.~Fernandez-Granda.
\newblock Super-resolution from noisy data.
\newblock {\em Journal of Fourier Analysis and Applications}, 19(6):1229--1254,
  2013.

\bibitem{Candes_2014}
E.~J. Cand{\`e}s and C.~Fernandez-Granda.
\newblock Towards a mathematical theory of super-resolution.
\newblock {\em Communications on Pure and Applied Mathematics}, 67(6):906--956,
  2014.

\bibitem{Chizat_2018}
L.~Chizat and F.~Bach.
\newblock On the global convergence of gradient descent for over-parameterized
  models using optimal transport.
\newblock {\em arXiv preprint arXiv:1805.09545}, 2018.

\bibitem{ciarlet1989introduction}
P.~G. Ciarlet, B.~Miara, and J.-M. Thomas.
\newblock {\em Introduction to numerical linear algebra and optimisation}.
\newblock Cambridge University Press, 1989.

\bibitem{Castro_2015}
Y.~De~Castro, F.~Gamboa, D.~Henrion, and J.-B. Lasserre.
\newblock Exact solutions to super resolution on semi-algebraic domains in
  higher dimensions.
\newblock {\em arXiv preprint arXiv:1502.02436}, 2015.

\bibitem{Duval_2017}
V.~Duval, P.~Catala, and G.~Peyr{\'e}.
\newblock A low-rank approach to off-the-grid sparse deconvolution.
\newblock {\em preprint}, 2017.

\bibitem{Duval_2015}
V.~Duval and G.~Peyr{\'e}.
\newblock Exact support recovery for sparse spikes deconvolution.
\newblock {\em Foundations of Computational Mathematics}, 15(5):1315--1355,
  2015.

\bibitem{Duval_2017thin1}
V.~Duval and G.~Peyr{\'e}.
\newblock Sparse regularization on thin grids i: the lasso.
\newblock {\em Inverse Problems}, 33(5):055008, 2017.

\bibitem{Duval_2017thin2}
V.~Duval and G.~Peyr{\'e}.
\newblock Sparse spikes super-resolution on thin grids ii: the continuous basis
  pursuit.
\newblock {\em Inverse Problems}, 33(9):095008, 2017.

\bibitem{Golbabaee_2018}
M.~Golbabaee and M.~E. Davies.
\newblock Inexact gradient projection and fast data driven compressed sensing.
\newblock {\em IEEE Transactions on Information Theory}, 2018.

\bibitem{Golub_2012}
G.~H. Golub and C.~F. Van~Loan.
\newblock {\em Matrix computations}, volume~3.
\newblock JHU press, 2012.

\bibitem{Gribonval_2017}
R.~Gribonval, G.~Blanchard, N.~Keriven, and Y.~Traonmilin.
\newblock {Compressive Statistical Learning with Random Feature Moments}.
\newblock {\em Preprint}, 2017.

\bibitem{Keriven_2016}
N.~Keriven, A.~Bourrier, R.~Gribonval, and P.~P{\'e}rez.
\newblock {Sketching for Large-Scale Learning of Mixture Models}.
\newblock Preprint, 2016.

\bibitem{Keriven_2017}
N.~Keriven, N.~Tremblay, Y.~Traonmilin, and R.~Gribonval.
\newblock Compressive k-means.
\newblock In {\em Acoustics, Speech and Signal Processing (ICASSP), 2017 IEEE
  International Conference on}, pages 6369--6373. IEEE, 2017.

\bibitem{Ling_2017}
S.~Ling and T.~Strohmer.
\newblock Regularized gradient descent: a non-convex recipe for fast joint
  blind deconvolution and demixing.
\newblock {\em Information and Inference: A Journal of the IMA}, 2017.

\bibitem{Poon_2018}
C.~Poon, N.~Keriven, and G.~Peyr{\'e}.
\newblock Support localization and the fisher metric for off-the-grid sparse
  regularization.
\newblock {\em arXiv preprint arXiv:1810.03340}, 2018.

\bibitem{Sriperumbudur_2010}
B.~K. Sriperumbudur, A.~Gretton, K.~Fukumizu, B.~Sch{\"o}lkopf, and G.~R.
  Lanckriet.
\newblock Hilbert space embeddings and metrics on probability measures.
\newblock {\em Journal of Machine Learning Research}, 11(Apr):1517--1561, 2010.

\bibitem{Tang_2013}
G.~Tang, B.~N. Bhaskar, P.~Shah, and B.~Recht.
\newblock Compressed sensing off the grid.
\newblock {\em IEEE transactions on information theory}, 59(11):7465--7490,
  2013.

\bibitem{Timan_1963}
A.~F. Timan.
\newblock {\em Theory of approximation of functions of a real variable}.
\newblock International Series of Monographs on Pure and Applied Mathematics,
  vol. 34, Pergamon Press, Oxford, 1963.

\bibitem{Traonmilin_2017}
Y.~Traonmilin, N.~Keriven, R.~Gribonval, and G.~Blanchard.
\newblock Spikes super-resolution with random fourier sampling.
\newblock In {\em SPARS workshop 2017}, 2017.

\bibitem{Waldspurger_2018}
I.~Waldspurger.
\newblock Phase retrieval with random gaussian sensing vectors by alternating
  projections.
\newblock {\em IEEE Transactions on Information Theory}, 2018.

\bibitem{Zhao_2015}
T.~Zhao, Z.~Wang, and H.~Liu.
\newblock A nonconvex optimization framework for low rank matrix estimation.
\newblock In {\em Advances in Neural Information Processing Systems}, pages
  559--567, 2015.

\end{thebibliography}

\end{document}